\documentclass{article}
\pdftrailerid{}
\usepackage{amsmath}
\usepackage{amssymb}
\usepackage{amsthm}
\usepackage{graphicx} \usepackage{subcaption}
\usepackage{geometry}
\usepackage{parskip}
\usepackage[dvipsnames]{xcolor}
\usepackage{listings}
\usepackage{tcolorbox}
\usepackage{needspace}
\usepackage[numbers]{natbib}
\usepackage[hidelinks]{hyperref}
\usepackage{cleveref}

\newtheorem{theorem}{Theorem}[section]
\newtheorem{proposition}[theorem]{Proposition}
\DeclareMathOperator{\KL}{KL}
\usepackage{authblk}

\title{
\Large
Skeleton Coding: Program Structure as a Specification for Agents \\
}
\author{Aditya Golatkar$^{*}$\hspace{1.6em}Alessandro Achille$^{*}$\hspace{1.6em}Renos Zabounidis$^{*}$\hspace{1.6em}Vlad Sobal$^{*}$\hspace{1.6em}Yangruibo Ding\\[0.35em]
Baishakhi Ray\hspace{1.6em}Jonathan Buck\hspace{1.6em}Guy Van den Broeck\hspace{1.6em}Maryam Fazel\hspace{1.6em}Pietro Perona\hspace{1.6em}Wei Xia\\[0.35em]
Stefano Soatto}
\affil{AWS Agentic AI}
\date{}

\renewcommand{\paragraph}[1]{\textbf{#1}}

\definecolor{lstbg}{RGB}{248,248,250}
\definecolor{lstframe}{RGB}{200,200,210}
\definecolor{lstkw}{RGB}{0,92,197}
\definecolor{lstears}{RGB}{140,30,160}
\definecolor{lstsect}{RGB}{0,128,110}
\definecolor{lsthead}{RGB}{160,60,20}
\definecolor{lststr}{RGB}{3,106,7}
\definecolor{lstcmt}{RGB}{110,110,120}
\definecolor{skeletonbg}{RGB}{255,255,255}
\definecolor{skeletonframe}{RGB}{42,120,214}
\newsavebox{\figureoneplots}
\lstdefinelanguage{ears}{
    morekeywords=[3]{WHEN,IF,THEN,THE,SYSTEM,SHALL,WHERE,WHILE,CALLABLE,CALLER},
    morecomment=[l][\color{lsthead}\bfseries]{\#},
    sensitive=true,
}
\lstdefinelanguage{pyi}{
    morekeywords={class,def,self,None,True,False},
    morekeywords=[2]{Args,Returns,Raises,Ensures,Requires,Invariants},
    morekeywords=[3]{WHEN,IF,THEN,THE,SYSTEM,SHALL,WHERE,WHILE,CALLABLE,CALLER},
    sensitive=true,
}
\lstdefinelanguage{rustmini}{
    morekeywords={pub,fn,struct,impl,const,self,Self,let,mut,for,in,return,use,trait,enum,match,where},
    morecomment=[l]{//},
    morestring=[b]",
    sensitive=true,
}

\begin{document}

\maketitle
{\renewcommand{\thefootnote}{}\footnotetext{$^{*}$Equal contribution sorted alphabetically}\footnotetext{Corresponding authors: \texttt{agolatka@amazon.com}, \texttt{aachille@amazon.com}}}

\begin{abstract}
    Coding agents are typically provided end-to-end behavioral specifications that the code they produce should satisfy, such as tests and input--output examples. How many such specifications must a developer provide, and which ones, for the agent to produce the developer's intended program? We argue that this cannot be determined a priori, for it depends on the \emph{structure} of the implementation, which is left to the agent to decide. An agent first commits to a structure---a decomposition into data types, modules, and functions---which fixes a class of candidate programs; behavioral specifications then select a program within that class. A set of tests that suffices for one structure may leave another badly underdetermined, requiring different edge cases or exponentially more examples. Since the structure is chosen by the agent, developers cannot ensure their specification is complete without reverse-engineering it from thousands of lines of generated code. To resolve this, we introduce \textbf{Skeleton Coding}, in which developers jointly specify the structure---in high-level, code-like format---alongside behavioral specifications.
    Theoretically, we model AI-assisted coding as a control problem, and show that skeletons can reduce exponentially the number of tests developers must supply to steer an agent toward an acceptable program.  Empirically, across four models and three benchmarks, the best-performing skeleton specifications reduce mean test error by 48.4\% relative to the best end-to-end behavioral specifications. Even when the agent can query a complete reference implementation, adding a skeleton reduces mean test error relatively by up to 27.0\%. Skeletons also narrow the average gap between small and large models by about 52.4\% while lowering implementation cost by 15--26$\times$, and reduce code erosion by 24.7\% on average.
\end{abstract}

\begin{figure}[t]
    \centering
    \captionsetup[subfigure]{font=scriptsize,skip=2pt}
    \sbox{\figureoneplots}{\includegraphics[width=0.58\textwidth]{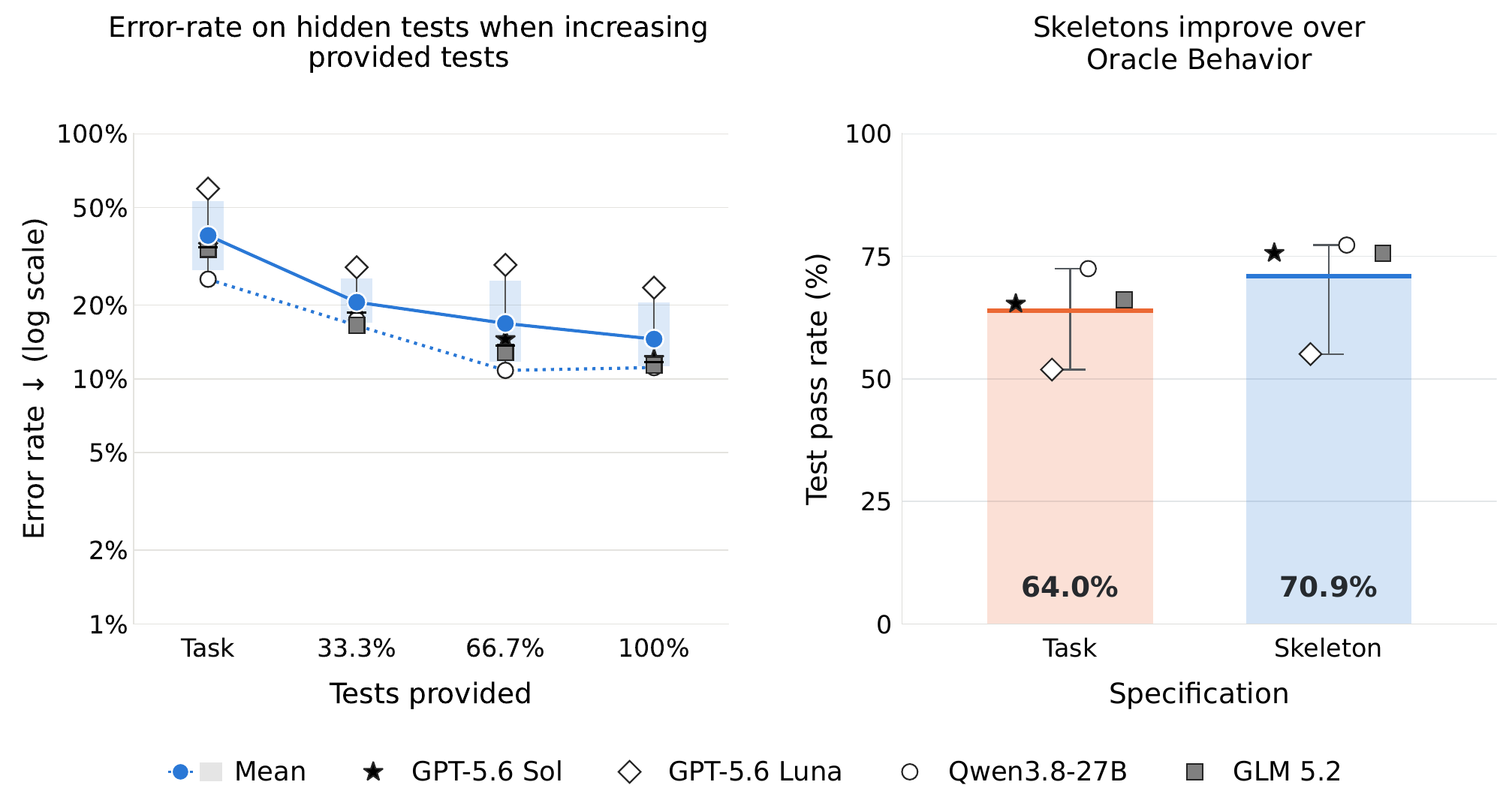}}
    \begin{minipage}[t]{0.58\textwidth}
        \vspace{0pt}
        \setlength{\parskip}{0pt}
        \usebox{\figureoneplots}
        \par\nointerlineskip
        \begin{subfigure}[t]{0.49\linewidth}
            \caption{}
            \label{fig:oracle_tests_scaling_a}
        \end{subfigure}\hfill
        \begin{subfigure}[t]{0.49\linewidth}
            \caption{}
            \label{fig:oracle_skeleton_comparison}
        \end{subfigure}
    \end{minipage}\hfill
    \begin{subfigure}[t]{0.40\textwidth}
        \vspace{0pt}
        \begin{minipage}[t][\ht\figureoneplots][t]{\linewidth}
        \setlength{\parskip}{0pt}
        {\centering\sffamily\fontsize{6}{7}\selectfont Skeletons with local contracts\par}
        \vspace{3pt}
        \begin{tcolorbox}[colback=skeletonbg,colframe=skeletonframe,
            boxrule=0.4pt,arc=1.3mm,boxsep=0pt,
            left=4pt,right=4pt,top=4pt,bottom=4pt,
            before skip=0pt,after skip=0pt]
\begin{lstlisting}[language=pyi,basicstyle=\ttfamily\fontsize{5.5}{6.5}\selectfont\color{black},columns=fixed,basewidth=0.5em,frame=none,backgroundcolor=\color{skeletonbg},xleftmargin=0pt,xrightmargin=0pt,aboveskip=0pt,belowskip=0pt]
class View:
    buffer: list[float]; offset: int; stride: int; size: int

def slice_view(v: View, lo: int, hi: int) -> View:
    """Return a slice that shares the original data buffer.

    Requires:
        THE CALLER SHALL ensure 0 <= lo <= hi <= v.size.
    Ensures:
        THE CALLABLE SHALL return r with:
            r.buffer is v.buffer; r.stride == v.stride,
            r.offset == v.offset + lo * v.stride,
            r.size == hi - lo.
        THE CALLABLE SHALL NOT modify v or its buffer.
    """
    ...
\end{lstlisting}
        \end{tcolorbox}
        \end{minipage}
        \caption{}
        \label{fig:local_contract_skeleton}
    \end{subfigure}
    \captionsetup{skip=4pt}
    \caption{\textbf{Behavioral specifications do not pinpoint correct programs.} \textbf{(Left)} We ask a coding agent to implement a code base from ProgramBench given its documentation and a fraction of its test suite. We then measure error rate across the hidden suite. On average, no number of tests provided results in a program that passes all hidden tests.
    \textbf{(Middle)} \textit{Even providing a full behavioral specification is not enough:} We measure the pass rate with access to an oracle providing the full and exact input-output behavior. Even so, adding a skeleton with local contracts increases the pass rate averaged across models by 10.9\% relative to the same setup without a skeleton on ProgramBench. \textbf{(Right)} 
    A skeleton specifying the local contract for slicing a strided view.}
    \label{fig:oracle_tests_scaling}
\end{figure}

\begin{figure}[t]
    \centering
    \includegraphics[width=\textwidth]{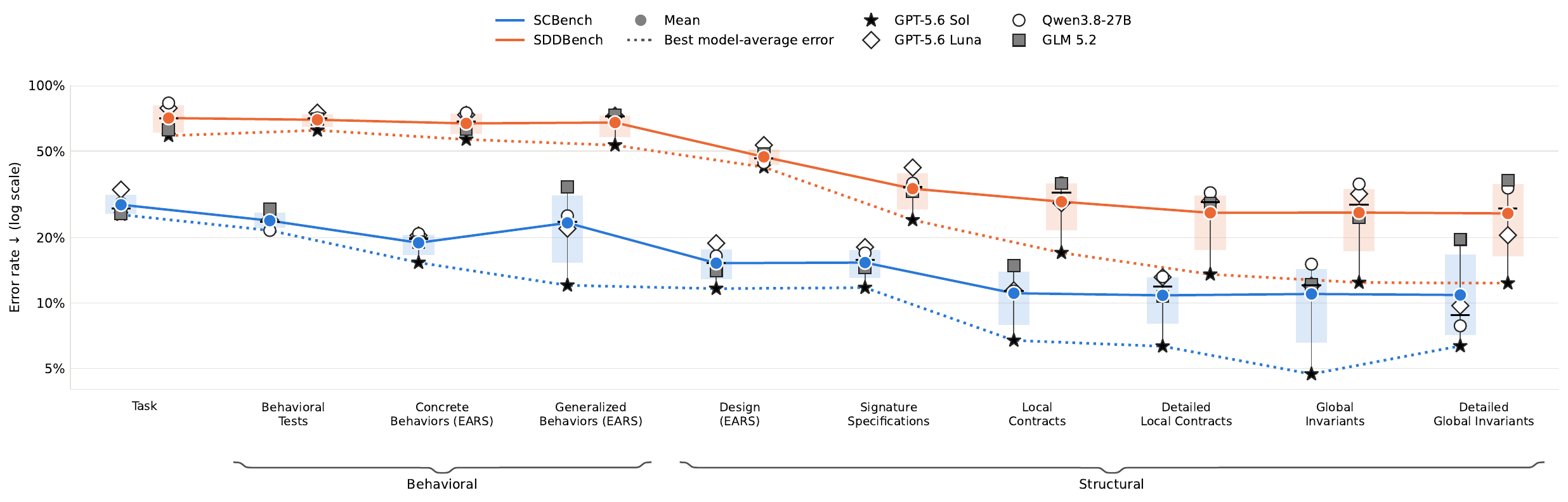}
    \caption{\textbf{Value of different types of specifications.}
    We generate specifications of different kinds and measure the error rate of a coding agent using them. No matter the format, behavioral specifications that describe only end-to-end behavior
    underperform specifications that describe the developer's intended structure. Of the structural specifications, code-like specifications that describe function signatures and invariants outperform design files in natural language. Adding overly detailed specifications decreases the accuracy of smaller models.
    }
    \label{fig:pass_rate}
\end{figure}

\section{Introduction}
\label{sec:introduction}

Coding agents are shifting the role of developers from writing code to \textit{controlling} an agent into producing an acceptable code base. This raises a key question: how can developers specify their intent and safely guide the agent toward producing an acceptable program, while minimizing their effort?

A common ideal has been to move toward the separation of concerns~\cite{dijkstra1982role} offered by specification-driven development~\cite{kiro2025}:
The developer's job should not be to dictate \textit{how} to implement the code base, but to specify \textit{what} behavior any acceptable result must have --- through prompts, examples, tests, or formal contracts --- and verify that the result fulfills that intent. 

Ideally, this separation of concerns between the \textit{what} and the \textit{how} promises to free developers from the laborious task of writing code. However, we show that it can actually \textit{exponentially increase} the effort required.

The separation fails because how many tests are needed, and which ones, is not a property of the function being implemented but of the \textit{structure} of its implementation.
Consider for example implementing the indexing rules of NumPy: fully specifying all their edge cases takes many tests and examples, yet most of them follow naturally from a single structural choice --- implementing arrays as strided views of memory~\cite{vanderwalt2011numpy,harris2020numpy} --- which no example or end-to-end test can express.
Yet it is the agent, not the developer, that chooses the structure. A test suite that suffices for one design can leave another badly underdetermined, so a developer cannot know whether their tests are sufficient without further extensive manual testing, or inspection of thousands of lines of generated code.

Based on this, we argue that the most effective control action for developers is to partially break the separation, and jointly specify structure and behavior.
To this effect, we introduce \textbf{Skeleton Coding}, in which developers specify the structure of an application directly, in high-level, code-like form: data types, module boundaries, function signatures, data flow, per-function contracts, and invariants --- leaving the actual implementation to the agent (\Cref{sec:spec-examples}). These are the decisions where skilled developers traditionally add the most value, and that end-to-end tests cannot express.

Formally, the coding agent faces an identification problem~\cite{astrom1971identification,gold1967identification,liu2016teaching}. Let $\mathcal H$ be the programs the agent can produce and $\mathcal G\subseteq\mathcal H$ the developer's goal set of acceptable programs. Given input--output observations $D$ of a reference program, the agent must find a program in $\mathcal G$. It first commits to a design---a decomposition into data types, modules, and functions---which restricts the search to a class $\mathcal H_S$, and tests then select a program within it. Even when the training of the model results in a strong prior, some tasks require exponentially many observations before the agent succeeds with high probability. \Cref{fig:oracle_tests_scaling}~(left) illustrates the problem: supplying more tests raises the pass rate, yet on average no number of tests yields a program that passes the full suite. From the developer's side, this is a control problem: find the cheapest specification $D$ that makes the agent sample a program from $\mathcal G$ (\Cref{sec:method}). Unfortunately, the optimal $D$ depends on $\mathcal H_S$, which the developer does not see.

Skeleton coding addresses the issue by letting 
the developer supply $\mathcal H_S$.
Instead of carving out the goal set one test at a time, developers can shrink the hypothesis class and provide the optimal tests for it. This can cut the number of needed tests exponentially, for several reasons. 
First, skeletons can quickly reduce the dimension of the hypothesis space. But, less intuitively, even when two skeletons lead to search spaces of the same size, a well-chosen hypothesis class may need exponentially fewer tests to identify the correct program than a poorly chosen one (\Cref{thm:structural_end_to_end_separation}).  Second, with end-to-end tests alone the program's internal state is hidden: an error in one component is visible only if it propagates to the output. 
By exposing intermediate interfaces, a skeleton makes that state observable, so for example a dataflow composing $n$ function can be identified with $O(n)$ local tests, whereas end-to-end tests may need $O(2^n)$ (\Cref{thm:blue_local_separation}).
Third, some requirements, such as internal implementation details, are not observable from outputs at all and no number of tests can enforce them. A skeleton can instead state them directly (\Cref{prop:blue_behavioral_obstruction}).

Developers rarely begin with a complete specifications; they adapt it to correct issues and steer the agent toward the goal. However, since prompting a model does not perform Bayesian conditioning~\cite{falck2024martingale}, there is no guarantee that providing additional specifications will eventually converge to the goal.
However, adding formally verifiable contracts~\cite{leino2010dafny,lattuada2023verus,demoura2021lean4} to the skeletons circumvent the issue (\Cref{prop:component_rejection}). Even when informal contracts are not used, an imperfect verifier can still ensure faster convergence to a valid implementation (\Cref{prop:judge}).

Because skeletons resemble code, teams can version, diff, review, and revise them alongside the implementation as requirements evolve. A persistent skeleton also records which decisions are requirements, so agents can rewrite cleanly when adding functionality instead of accreting patches to preserve accidental behavior, a major contributor to ``AI code slop'' on long-lived code bases~\cite{orlanski2026slopcodebench}.

\paragraph{Empirical results.}
To show the value of skeletons and structural specifications, we compare the pass-rate with and without skeletons across various coding benchmarks and architectures. On SCBench~\cite{orlanski2026slopcodebench} and SDDBench~(\Cref{app:sddbench}), structural specifications raise mean pass rate from 72\% to 89\% and from 29.0\% to 74.2\%, whereas even the best purely behavioral specifications gain at most 9.4 percentage points (\Cref{fig:pass_rate}). Surprisingly, on ProgramBench~\cite{yang2026programbench}, where the coding agent is provided with a \textit{full} behavioral-specifications (an oracle binary), skeletons still raise relative pass rate by 10.9\% (\Cref{fig:oracle_tests_scaling}, right).
Skeletons also narrow the gap between small and large models~\cite{openai2026gpt56}: with a detailed skeleton, a small model comes within 5--8 percentage points of the largest model's pass rate at 15--26$\times$ lower implementation cost (\Cref{fig:cost_performance}).
Aside from outperforming behavioral specifications, our proposed code-like skeletons outperform structural specifications provided as prose design documents (\Cref{fig:architecture_scaling}).
In terms of code maintainability, skeletons reduce code erosion by 24.7\% where behavioral specifications increase it (\Cref{fig:erosion}). In iterative repair, they reach similar error rates with 32\% fewer queries to the reference executable (\Cref{fig:brownfield_skeleton_coding}).

\setcounter{topnumber}{1} \section{Controlling Coding Agents with Skeletons}
\label{sec:method}

\begin{figure}[t]
    \centering
    \includegraphics[width=0.8\textwidth]{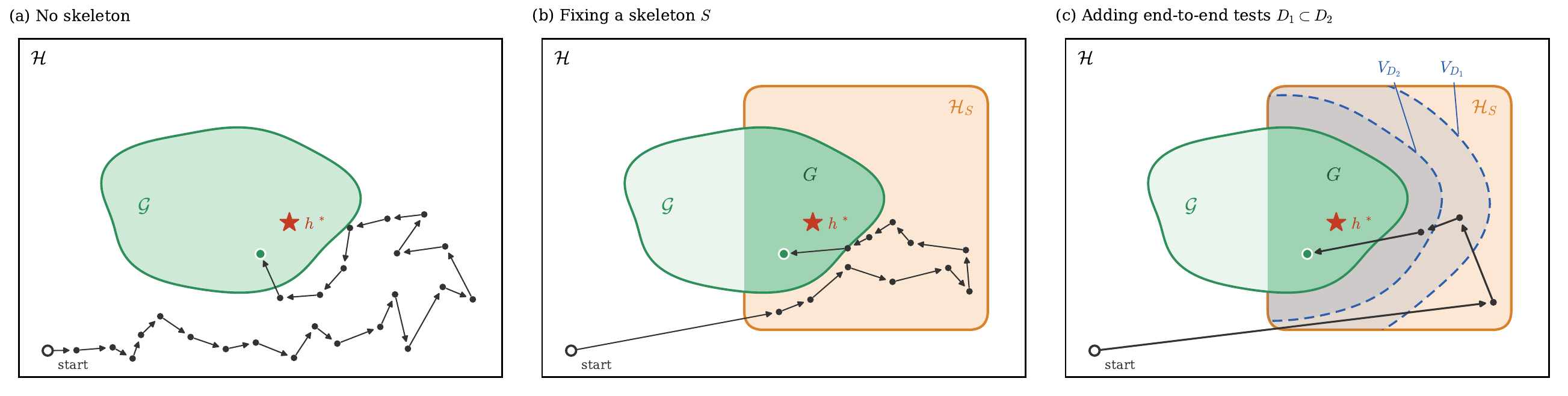}
    \caption{\textbf{Skeletons and tests shrink the set of programs the agent considers}:
    (a)~Without a skeleton, the agent's prior spreads over all programs. (b)~Fixing a
    skeleton restricts the agent to programs with that skeleton. (c)~Each set of
    end-to-end tests further keeps only the programs that agree with the reference on
    the tested inputs. 
    }
    \label{fig:set_overview}
\end{figure}

Let $\mathcal H$ be the set of all programs the agent can produce. Let $\mathcal G\subseteq\mathcal H$ be a subset of programs that the developer would consider an \textit{acceptable} realization of their intent. The goal of the developer is to control the agent into producing a program in $G$. To do so, the developer can prompt the agent with specifications describing their intent.

The simplest form of specifications are \textbf{behavioral specifications}, which describe the expected input--output
behavior of programs in $G$ through a set of examples
$D = \{(x_i, y_i)\}_{i=1}^N$. Note that end-to-end tests are a special case with finitely many examples, while property testing~\cite{claessen2000quickcheck} and 
formal contracts~\cite{meyer1992designbycontract} can be viewed as specifications that implicitly
characterize infinitely many such pairs. For clarity, in the following we focus on the case where $D$ is finite and given by end-to-end tests.

A good set $D$ of tests is such that, after prompting the agent with $D$, it will generate with high probability a program in $G$. That is, letting $p_0(h)$ denote the agent's distribution over programs after prompting with the task alone, we want:
\[
p_0(G|D) > 1 - \delta
\]
for a small $\delta>0$. However, producing tests $D$ is expensive for the developer. Hence, we are interested in the minimum number of tests needed:
\begin{equation}
    N_\delta=\min\{|D|:p_0(\mathcal G\mid D)\ge1-\delta\},
\end{equation}
where the minimum is over all test suites~\cite{goldman1995teaching,liu2016teaching,angluin1988queries}.

A program, however, can also be specified by describing its architecture. We call a \textbf{skeleton} $S$ as a partial abstraction of a program, describing its data types, module boundaries, permitted calls and data flow, and function
declarations with signatures, but omits their implementation (see \Cref{fig:oracle_tests_scaling}).

Specifying a skeleton $S$ restricts the search space to the subset $H_S \subseteq H$ of
programs sharing the same skeleton. In turn, this can drastically change the number of tests required to generate a program in $G$. 
Let $p_S=p_0(\cdot \mid\mathcal H_S)$ and denote the agent distribution restricted to by $H_S$ (\Cref{fig:set_overview}b), the minimum number of tests required when using the skeleton $S$ is:
\begin{equation}\label{eq:test_complexity}
    N_\delta(S)=\min\{|D|:p_S(\mathcal G\mid D)\ge1-\delta\},
\end{equation}
The rest of this section shows that appropriately specifying $S$ can, with little developer effort, exponentially reduce $N_\delta(S)$.

But first, what happens when the developer does not specify a skeleton? The agent will have to select an architecture based on $D$ alone. Let $p_0(S|D)$ denote the probability of the model selecting a given architecture, then the probability of success is an average over the agent's choices:
\begin{equation}\label{eq:goal_mixture}
    p_0(\mathcal G\mid D)=\sum_S p_0(S\mid D)\,p_S(\mathcal G\mid D).
\end{equation}
Tests act on both factors: they select a program within each skeleton, and
they reweigh the skeletons by the prior mass of their surviving programs. A
developer who supplies a skeleton $S$ removes the agent's choice and
replaces the average with a single term. Even when \textit{no tests} are specified, this can still raise the
success probability exactly when the success rate of the developer chosen skeleton $p_S(\mathcal G)$ exceeds the agent's average $p_0(\mathcal G)$.

\subsection{Structure determines what tests are needed}
\label{sec:structure_tests}

Specifying a skeleton $S$ restricts the hypothesis class $H_S$ the coding agent needs to consider. By itself, this can quickly reduce the number of tests $N_\delta(S)$ needed to reach the goal (\Cref{fig:set_overview}) by simply reducing the search space (when $H_S$ intersects with $\mathcal{G}$).
However, even when the two skeletons admit the same number of programs, a developer can reduce exponentially the number of tests needed by designing the skeleton smartly, as the following theorem shows.

\begin{theorem}[Exponential separation between skeletons]
\label{thm:structural_end_to_end_separation}
For every $d\ge1$, there are a target function and two skeletons
$S_A,S_B$ for it, each admitting $2^d$ programs with distinct behaviors
under a uniform prior, such that for every $0\le\delta<1/2$
\[
    N_\delta(S_A)=2^d-1,\qquad N_\delta(S_B)=d.
\]
Moreover, a suite of $d$ tests that guarantees success under $S_B$ would lead to success probability only $p(G|S_A) = 1/(2^d-d)$ under $S_A$.
\end{theorem}

\begin{figure}[t]
    \centering
    \includegraphics[width=0.7\textwidth]{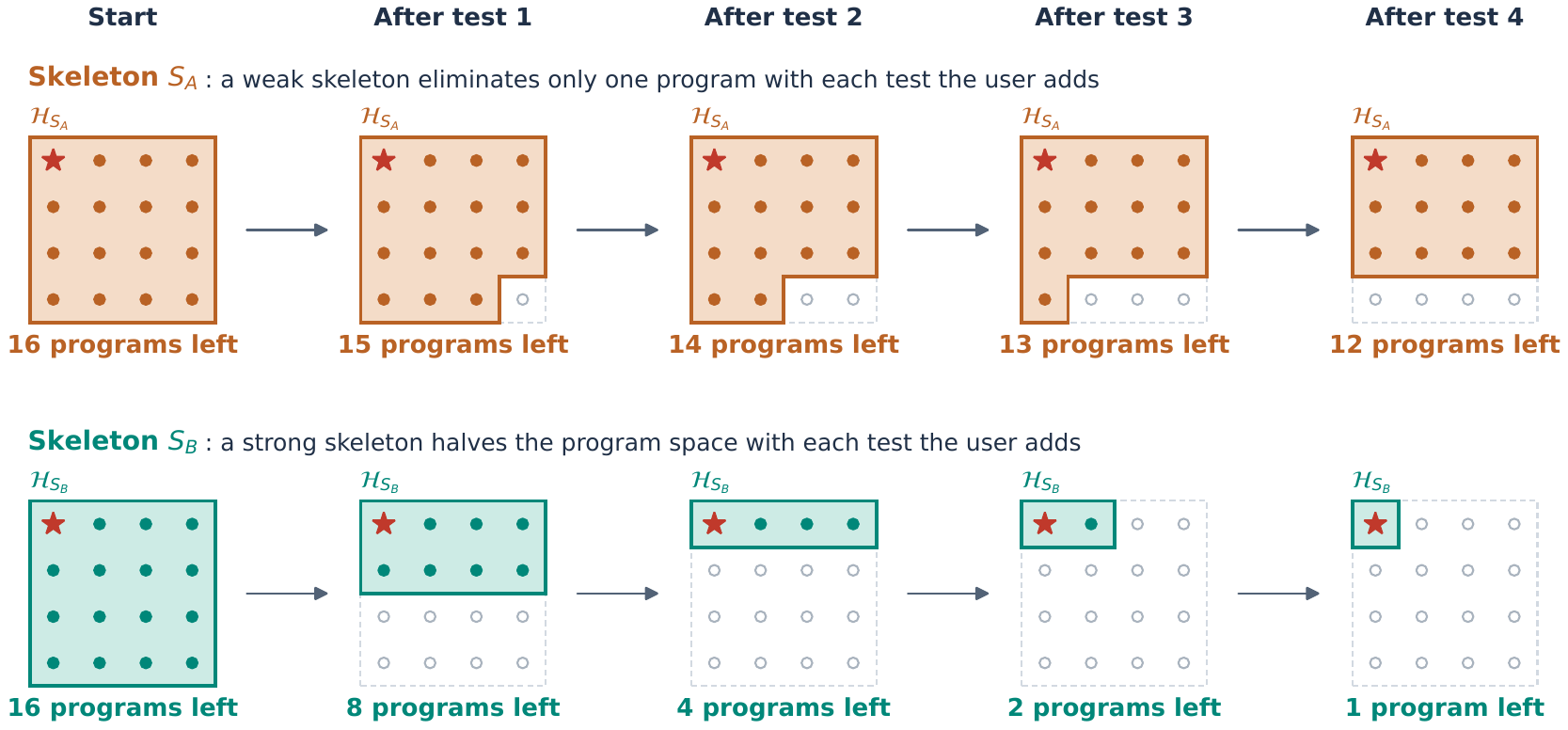}
    \caption{\textbf{A strong skeleton makes better use of tests.} Visualization of how tests affect the two skeletons of \Cref{thm:structural_end_to_end_separation}. }
    \label{fig:skeleton_test_efficiency}

\end{figure}

The two skeletons (described in the Appendix) differ in how \textit{informative} end-to-end tests are. Under $S_A$, each test excludes at most a single program, while under
$S_B$, each test rules out
half of the incorrect programs (\Cref{fig:skeleton_test_efficiency}).

The second part of the theorem shows why a developer cannot determine the required test suite to guarantee success
without knowing the structure of the yet-to-be-written program. 
Suppose the developer creates an optimal test suite $|D| = d$ that guarantees success under the assumption the agent uses the architecture $S_B$. If the agent chooses $S_A$ instead, the same optimal test suite will fail with high probability.

One may wonder if providing $D$ designed for $S_B$ should not implicitly bias the agent toward picking $S_B$. But the reverse is actually true: by \Cref{eq:goal_mixture},
conditioning on $D$  shifts the agent toward the skeleton under which
they are uninformative, because that skeleton loses the fewest programs. In \Cref{fig:oracle_tests_scaling} (right), we show this empirically: even
when the agent has a full behavioral specification, supplying a
skeleton raises its pass rate, although the observations available to the agent are unchanged.

\subsection{Interfaces make internal state observable}
\label{sec:observability}

A second way in which skeletons reduce the number of tests needed is by changing what tests can observe~\cite{kalman1963mathematical}. An end-to-end test exposes a fault only
if the faulty code runs, corrupts the program state, and the corruption
propagates to the output~\cite{voas1992pie}. When later computation masks
the corruption, a faulty program returns the correct output on almost every
input, and end-to-end tests must search for the few inputs that reveal it.
The interfaces of a skeleton name the intermediate values, so each component
can be tested directly.

Consider for example a program that reads one input bit at each of $n$ stages and
carries one bit of state:
\[
    z_i=h_i(x_i,z_{i-1}),\qquad z_0=0,\qquad
    F_h(x_1,\ldots,x_n)=z_n,
    \qquad h_i:\{0,1\}^2\to\{0,1\}.
\]
Let $S$ be this skeleton. An end-to-end test observes only $F_h(x)$. However, having specified $S$, the developer can also specify local tests for each of the components $h_i$, and observe the intermediate hidden state $z_i$. The following theorem shows that this too reduces exponentially the number of tests the developer needs to specify.

\begin{theorem}[Exponential separation between local and end-to-end tests]
\label{thm:blue_local_separation}
Let the goal $\mathcal G\cap\mathcal H_S=\{h\in\mathcal H_S:F_h\equiv0\}$ be to implement the identically zero function. There is a prior over the transition tables, with
$p_S(\mathcal G)=1/2$, under which
\[
    N_\delta(S)\ge\left\lceil 2^n\,\frac{1-2\delta}{1-\delta}\right\rceil
    \qquad\text{for every }0\le\delta<1/2,
\]
and in particular $N_0(S)=2^n$. Conversely, under any prior giving mass to the correct solution, $4n$ local tests suffice to identify
$F_h$.
\end{theorem}

Note that the program has only $4n$ unknown bits (the truth table of each $h_i$), yet its output can hide
where a disagreement arose. In the proof, each incorrect program differs
from the reference on a single input~\cite{liu2016teaching}, which
end-to-end tests must find by search. A local test needs no such search,
because it observes a component before later computation can mask
it~\cite{freedman1991testability}. The skeleton provides this observability
through its interfaces and contracts. For example, in
\Cref{fig:oracle_tests_scaling} (left), we show that although the agent can
run the reference executable on any input, its pass rate approaches that of
the full test suite only as the evaluation tests themselves are supplied.
The bottleneck is finding the inputs that reveal an error.

\subsection{Some requirements are invisible to output tests}
\label{sec:unobservable}

The acceptable set $G$ could exclude programs that compute the same
function, but whose implementation the developer does not consider acceptable. For example, a developer may exclude implementations that use too many resources or have poorly maintainable code. However, those programs cannot be distinguished by end-to-end
tests that observes only the values a program returns as the next proposition quantifies.

\begin{proposition}[Behavioral tests preserve odds within a behavior class]
\label{prop:blue_behavioral_obstruction}
Let $h_0$ be a reference implementation, and let $[h_0] \subset H_S$ be the programs
in $\mathcal H_S$ that have the same input-output behavior. Suppose the goal set $G \subset [h_0]$ considers only a subset of those acceptable. Then, for every test set $D$:
\[
    p_S(\mathcal G\mid D)\le \frac{p_S(\mathcal G)}{p_S([h_0])},
\]
so $N_\delta(S)=\infty$ whenever $1-\delta$ exceeds this ratio.
\end{proposition}

In this situation, the only control input developers have to increase the chance of generating an acceptable program is to specify the skeleton $S$. Code quality is one such requirement. In \Cref{fig:erosion},
we show that structural specifications reduce erosion relative to the task-only condition, while behavioral
specifications can actually increase it.

\subsection{What makes a good skeleton}
\label{sec:skeleton_cost}

In \cref{sec:observability}, we have seen that skeletons can reduce the developer's effort by allowing specification of internal components. However, a good skeleton should also structure the implementation to reduce cross-component dependencies~\cite{parnas1972criteria}.
Let $h_1,\ldots,h_k$ be the
components of a program in $\mathcal H_S$. Let $N_i(G_i)$ be the smallest number of local tests to specify a behavior $G_i$ for the component $h_i$.
Even when each component has an acceptable behavior, their composition might not~\cite{benveniste2018contracts} and additional end-to-end tests need to be provided to pin-point a globally coherent implementation. Let $J_S(G_1,\ldots,G_k)$ be the smallest number of end-to-end tests that rule out unacceptable programs in
$\mathcal H_S\cap\prod_i G_i$.

\begin{theorem}[The cost of a skeleton splits into description, local tests, and compatibility]
\label{thm:blue_cost_decomposition}
For every skeleton $S$ with an acceptable reference
$h^\ast\in\mathcal H_S$ and $p_S(h)>0$ for every $h\in\mathcal H_S$,
where a guarantee means posterior correctness one,
\[
    \mathcal N_0(S)
    =
    \min_{G_1,\ldots,G_k}
    \Big[
         \underbrace{\textstyle \sum_i N_i(G_i)}_{\text{local tests}}
        +\underbrace{J_S(G_1,\ldots,G_k)}_{\text{compatibility}}
    \Big],
\]
where the minimum is over sets containing the reference components.
Moreover, $J_S(G_1,\ldots,G_k)=0$ if and only if
$\mathcal H_S\cap\prod_i G_i\subseteq G$.
\end{theorem}

The total effort required of the developer is the cost of writing the skeleton $S$, and of specifying the $\mathcal N_0(S)$ required tests.
What makes a good skeleton then? First, it should be short, stating only
the decisions that are needed to control the agent. Each of its components should be simple to
specify, with a narrow interface and a behavior that a few tests or a short
contract can pin down. Individually correct components should automatically lead to globally correct programs. For example, complete Hoare-style pre- and
post-conditions~\cite{hoare1969axiomatic} achieve $J_S=0$ (no compatibility tests are needed), since if each component satisfies its contract then the implementation is globally correct. This can come at the cost of a longer specification $S$. However, it might be still worth it: when $J_S=0$, components are independent so independent agents can implement them in parallel with no risk of conflict. Moreover, each agent then faces a simpler task, which a smaller model may handle (\Cref{fig:cost_performance}).

\begin{figure}[t]
    \centering
    \includegraphics[width=0.8\textwidth]{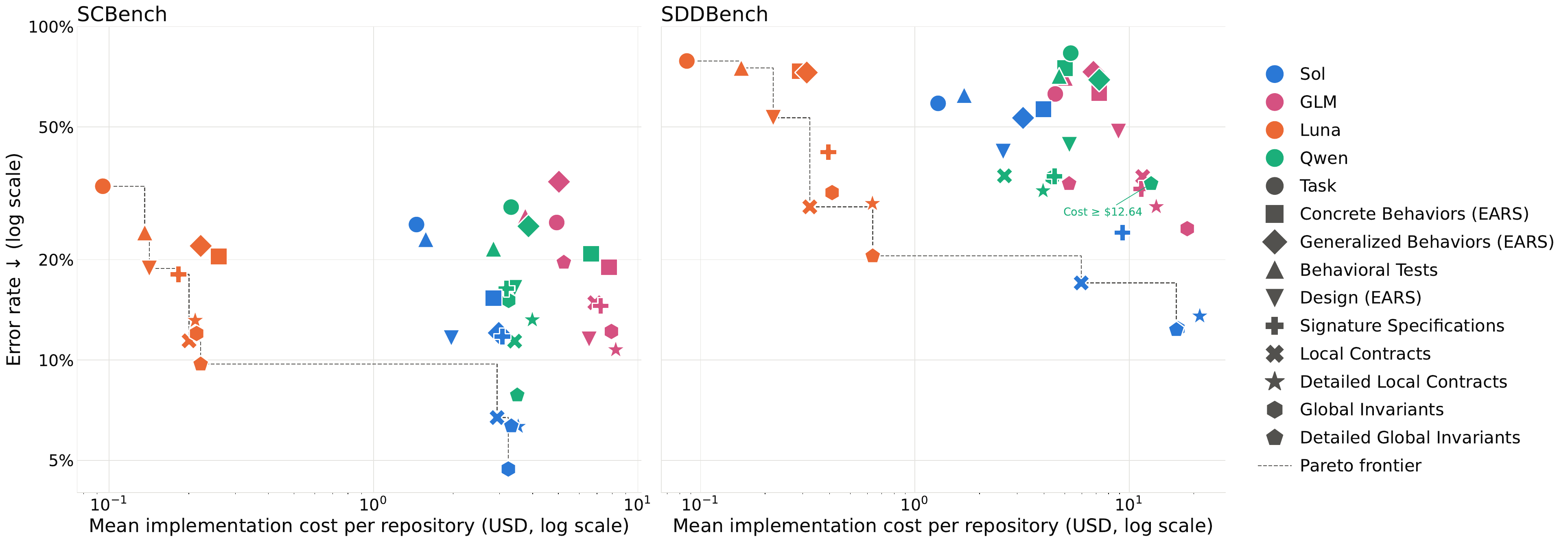}
    \caption{\textbf{Skeletons bridge the gap between large and small models}:
    Skeleton specifications significantly improve performance of smaller models, leading to accuracy similar to that of larger models at a lower cost.
    In particular, the entire cost--performance Pareto frontier (dashed) is formed mostly by small models with structural specifications.}
    \label{fig:cost_performance}
\end{figure}

\section{Closing the Loop: Checks and Refinement}
\label{sec:iterative_refinement}

\Cref{sec:method} describes an open-loop control setting: the specifications are provided as {\em initial conditions} to the agent. Initial conditions can at most condition, but not control, the execution. However, in practice LLMs are not even {\em conditioned} in the classical probabilistic sense, as much as they are {\em prompted}~\cite{falck2024martingale}. We will elaborate on this subtle but important difference, whose consequence is that it forces developers to iterate on a {\em slow loop}: they check each proposal and answer its failures.
Whether this loop works depends on whether the specification can be checked mechanically. In each case below, a skeleton turns a cost that multiplies
across components into one that adds. 

Let $E$ be a specification. So far, we have considered exact Bayesian conditioning $p(\cdot\,|\,E)$ on the specification which restricts the agent's
distribution to the programs $C_E=\{h:h\models E\}$ that satisfy it. Let $q_w(\cdot\mid T,E)$ be the
distribution actually produced by prompting the agent with the task $T$ and
the specification $E$, so that $p_0=q_w(\cdot\mid T)$.
Generally, the Bayesian posterior and the prompted model distribution will differ:
\[
p_0(\cdot\mid C_E) \neq q_w(\cdot\mid T,E).
\]
This means that a sample from $q_w(\cdot\mid T,E)$ might still violate $E$. Moreover, suppose the developer adds to an initial specification $E_0$ an additional specification $E_1$ to steer the agent toward the goal $G$ (such that $C_{E_1} \subseteq C_{E_0}$). The posterior entropy cannot increase in expectation over possible observations under the Bayesian predictive distribution, $\mathbb{E}_{E_1\mid E_0}[H(h\mid E_0,E_1)] \le H(h\mid E_0)$. Prompting with additional specifications, on the other hand, has no such guarantee. Indeed, in our experiments we show that overly detailed contracts can actually lower the pass rate (\Cref{fig:pass_rate}).

\textbf{Formal specifications.} Nonetheless, when the specification $E$ can be checked formally, the developer can close the loop by
rejection sampling, discarding proposals that fail the check. This ensures eventual convergence as more specifications are added, at the cost of additional compute to resample rejected programs. Local formal contracts can exponentially decrease the cost of rejection sampling as the following proposition shows. Let the
skeleton have components $i=1,\ldots,k$, let $G_i$ be the implementations of
component $i$ that satisfy its contract, and assume $G=G_1\times\cdots\times G_k$.
Given the skeleton, each component can be generated independently $h_i\sim q_i(\cdot\,|\,E_i)$, and we write $\alpha_i=q_i(G_i)$.
Let $q=q_1\times\cdots\times q_k$ be the resulting distribution over programs.

\begin{proposition}[Checking components separately]
\label{prop:component_rejection}
Resampling the whole program until every check passes takes
$\prod_i1/\alpha_i$ proposals in expectation, whereas resampling each
component until its own check passes takes $\sum_i1/\alpha_i$ component
proposals. Both return programs distributed as $q(\cdot\mid G)$.
\end{proposition}

\paragraph{Informal contracts.}
In our experiments we use informal contracts stated in the docstrings of the skeleton, checked by the agent itself
or by a model used as a judge. Suppose the judge accepts
every implementation consistent with its informal contract and let $\varepsilon_i$ be the
probability that it also accepts an implementation not in $G_i$.
This false-accept rate has two sources, $\varepsilon_i=\beta_i+(1-\beta_i)e_i$:
the fraction $\beta_i$ of violations that the informal contract, read literally,
still allows, and the rate $e_i$ at which the judge accepts the others.

\begin{proposition}[What an informal check guarantees]
\label{prop:judge}
Under the assumptions of \Cref{prop:component_rejection}:
\begin{enumerate}
\item[(i)] the probability that a component accepted by the judge actually satisfies the contract is\\
$\frac{\alpha_i}{\alpha_i+(1-\alpha_i)\varepsilon_i}$;
\item[(ii)] resampling each component until its check accepts yields a
program that satisfies every contract with probability at least
$1-\sum_i\varepsilon_i(1-\alpha_i)/\alpha_i$;
the probability is at most $\prod_i\alpha_i/\varepsilon$ for
whole-program rejection sampling with false-accept probability
$\varepsilon>0$.
\end{enumerate}
\end{proposition}

In particular, a check makes an accepted component likely correct only when  it errs much less often than the agent fails on that unit, $\varepsilon_i\ll\alpha_i$. The
 term $\beta_i$ is a floor that no judge removes: a violation that
the informal contract allows passes even a perfect judge. Making informal contracts more precise can drive $\beta_i$ to zero, leaving $\varepsilon_i$ equal to the judge's own error $e_i$; when that error also vanishes, we recover \Cref{prop:component_rejection}.

\section{Empirical Results}
\label{sec:experiments}

In this section, we empirically study how different kinds of specifications -- both behavioral and structural -- affect pass rate, code quality and implementation cost.

\paragraph{Types of specifications.} We test different specifications that describe \textbf{end-to-end behavior} \textit{but not the structure} of the code base: (a) \textit{Task Only}: a \textit{vibe-coding} style prompt describing the task; (b) \textit{Concrete Behaviors}: natural language prescription of observable behavior in EARS format~\cite{mavin2009ears}; (c) \textit{Behavioral Tests}: test suite on the public APIs. We then consider specifications describing the desired \textbf{structure} of the code base: (d) \textit{Design:} a prose document describes the system architecture~\cite{clements2010architectures}; (e) \textit{Function Stubs:} a code-like specification containing modules and function signatures with docstrings~\cite{vanrossum2014pep484}; (f) \textit{Local Contracts:} same as before, but adding local contracts, such as pre- and post-conditions to the docstring of each function~\cite{meyer1992designbycontract};
(g) \textit{Global Invariants:} same as before, adding properties about the composition of functions that should hold at the code-base level.

\paragraph{Models and Benchmarks.} We first consider greenfield applications, where the implementation is generated from scratch, using
programs drawn from SlopCodeBench (SCBench)~\cite{orlanski2026slopcodebench}, SDDBench~(\Cref{app:sddbench}), and
ProgramBench~\cite{yang2026programbench}. We evaluate GPT-5.6 Sol and
GPT-5.6 Luna~\cite{openai2026gpt56}, Qwen3.8-27B~\cite{qwen2026qwen38}, and GLM 5.2~\cite{zai2026glm52}. In each experiment, we take a reference program from one of these datasets and generate different types of
specifications from it. The implementation agent then receives the task
description and one specification and produces a reconstructed program without access to the reference implementation.
We run all agents in our experiments with the open-source AI Functions
framework~\cite{strandsaifunctions}\footnote{\url{https://github.com/strands-labs/ai-functions}}.

\paragraph{Metrics.} We measure agreement with evaluation tests and properties of the produced implementation.
\textbf{Test pass rate} is the fraction of evaluation tests that the implementation passes. It measures
agreement on the tested cases; passing a finite suite does not by itself establish membership in $\mathcal G$.
Code erosion and verbosity characterize implementation quality separately from behavioral correctness.
\textbf{Code erosion} \cite{orlanski2026slopcodebench} is the share of
cyclomatic-complexity~\cite{mccabe1976complexity} mass held by functions with cyclomatic complexity greater than 10, where each function's
mass is its cyclomatic complexity multiplied by the square root of its source-line count. It measures how much
complexity is concentrated in difficult-to-maintain code. \textbf{Verbosity} is the fraction of source lines flagged as
duplicated code or common code-slop patterns and measures unnecessary or repetitive implementation content.
\textbf{Implementation cost} is the estimated monetary cost of the implementation model's input, output, and cache-read
tokens, calculated using the corresponding per-token prices.

\textbf{Generation of specifications.} Each experiment starts from a ground-truth code base $h^\ast\in\mathcal G$. From this, based on the experimental setting $c$, we 
produce a specification $S_c$. The ground-truth code base is then kept hidden, and the
implementation agent reconstructs the program from the task and specification
alone, $h \sim q_\theta(\cdot\mid T,S_c)$. In the structural
conditions the specification is a skeleton; in the behavioral conditions it is a
set of observations derived from the reference; the task-only baseline
supplies no specification and samples from $p_0$. Pass rate therefore reflects both
the information each specification carries and the agent's ability to use that
information.

\setcounter{topnumber}{1}
\setcounter{totalnumber}{1}

\begin{figure}[t]
    \centering
    \includegraphics[width=0.48\textwidth]{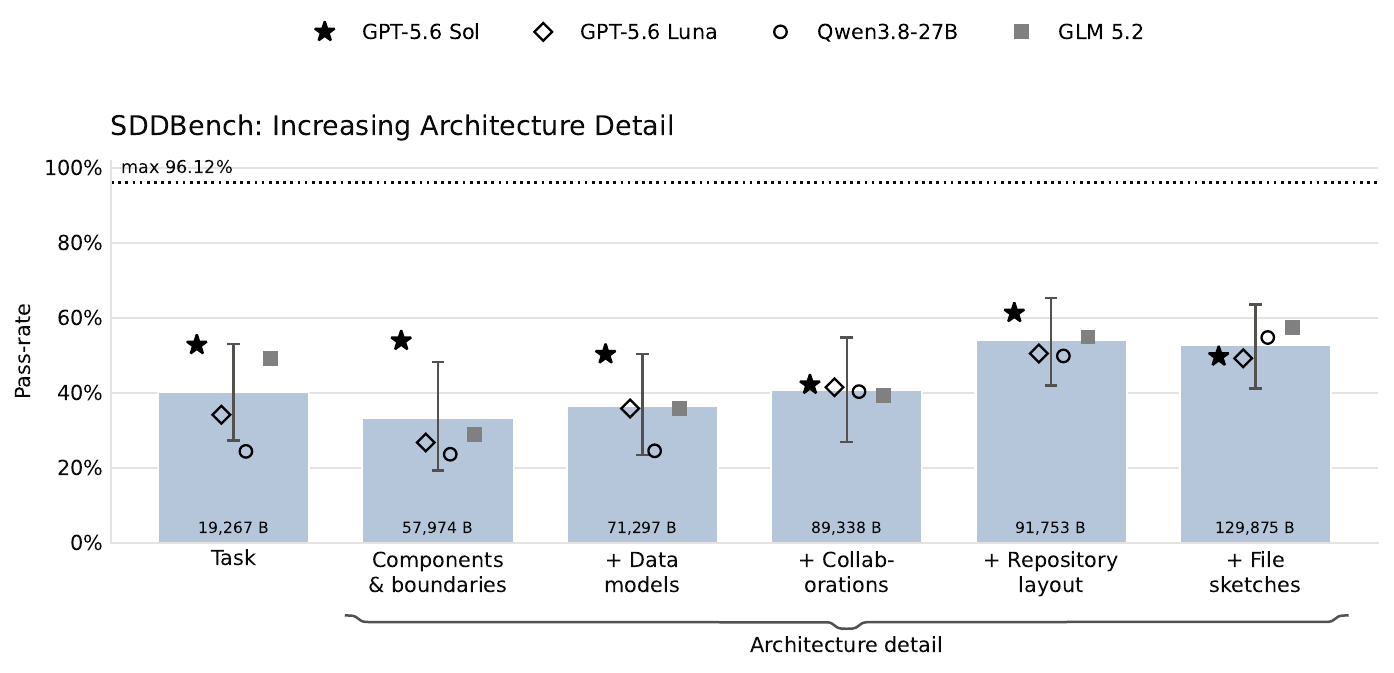}
    \hspace{2em}
    \includegraphics[width=0.33\textwidth]{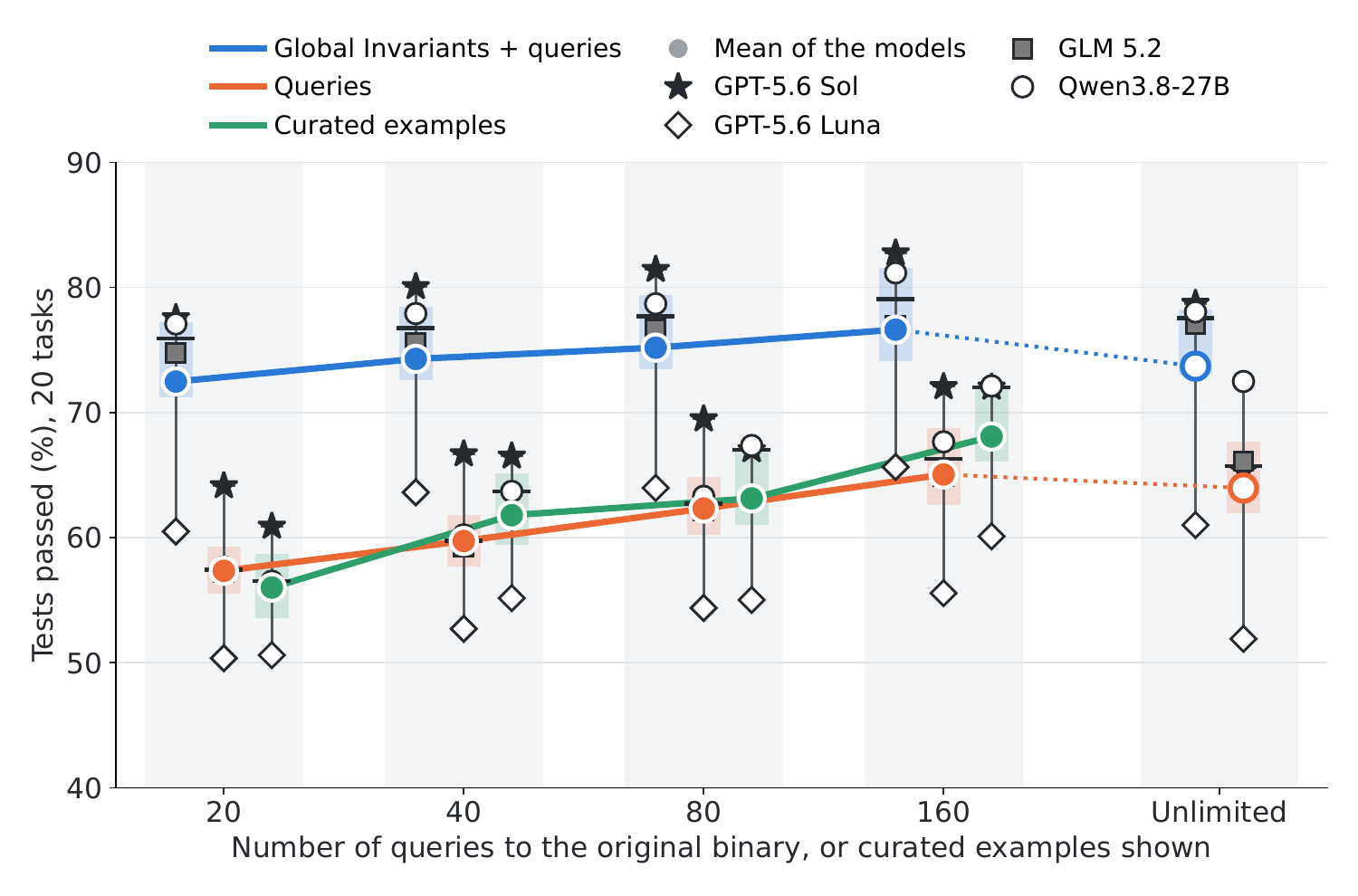}
    \caption{\textbf{(Left) More architecture detail does not match skeletons.}
    Adding architecture sections to the markdown file slightly raises mean pass rate, but the resulting pass rate remains below that achieved with skeleton-based specifications.
    \textbf{(Right) Accuracy improvement with additional behavioral observations.}
    On ProgramBench, we show the average test pass rate increases as we increase the number of observations provided to the model (unlimited lets the agent decide), but in all cases end-to-end behavioral specifications perform worse than skeletons with just 20 queries.
    }
    \label{fig:architecture_scaling}
\end{figure}

\subsection{Results}
\label{sec:results}

\paragraph{Behavioral testing alone does not identify valid programs}
First, in \Cref{fig:oracle_tests_scaling} we evaluate on ProgramBench how many behavioral specifications are needed to identify a program in the goal set. ProgramBench provides an oracle binary implementation for each task, which is in principle a full behavioral specification. However, even with access to it models only pass 61.5\% of the hidden tests. Still, this may be due to the agent not being able to make informative enough queries. To simulate a developer providing an informative set of observations, we provide the agent with an increasing number of the hidden tests. Surprisingly, while providing additional tests does help, on average no number of tests provided is enough to let the agent identify the correct program. This supports the theoretical observation that, with no structure specified, exponentially many tests might be required despite the strong prior of frontier coding agents.

\paragraph{Structural specifications outperform behavioral.} Next, we compare the performance of different types of specifications from behavioral descriptions and tests with accuracy from architecture documents and code-like skeletons. In \Cref{fig:pass_rate} we see that, no matter the format, behavioral specifications underperform compared to structural specifications. Of the latter, code-like specifications consisting of function stubs perform better than architectural design files in natural language. Adding additional structural constraints, such as local and global invariants, further improves the final pass rate. 
Across SCBench~\cite{orlanski2026slopcodebench} and SDDBench, structural specifications raise mean pass rate from 71.7\% to 89.2\% and from 29\% to 74.2\%, respectively, whereas behavioral specifications improve pass rate by at most 9.4 percentage points. Adding function stubs alone raises SDDBench mean pass rate from 53\% to 66.4\%. Together with Programbench, across 4 models the relative error rates reduces by 48.4\% from 40.7\% to 21.0\% (see \Cref{fig:pass_rate,fig:programbench_spec_types}) compared to the best end-to-end behavioral specifications.

\begin{figure}[t]
    \centering
    \includegraphics[width=\textwidth]{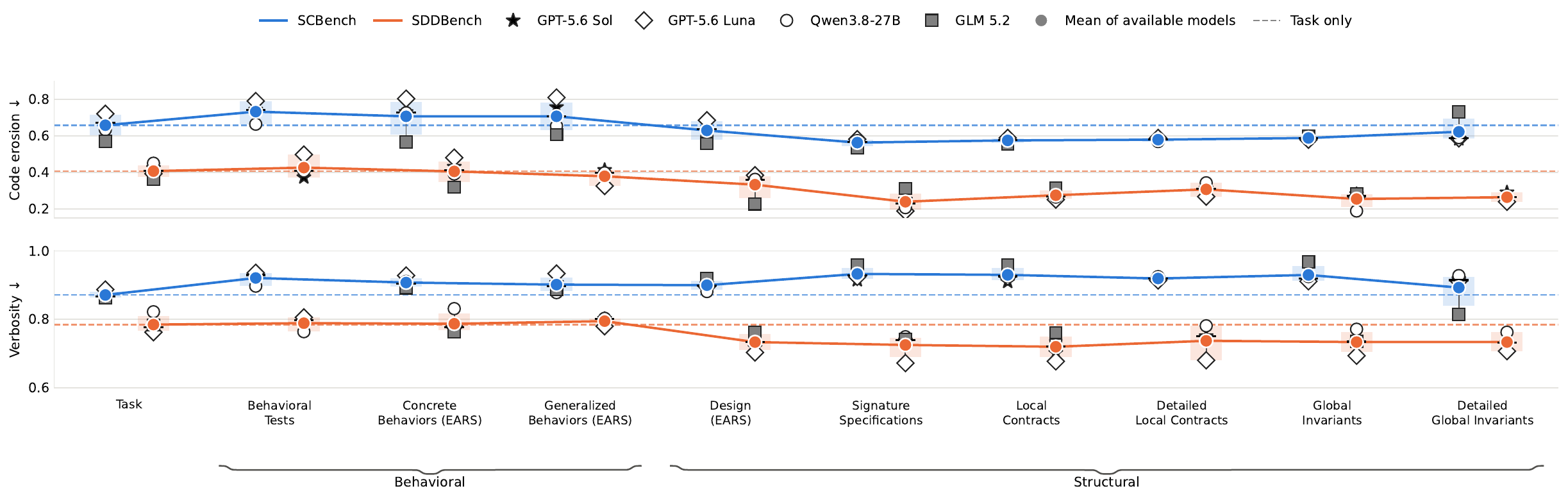}
    \caption{\textbf{Structural specifications produce less eroded code than behavioral specifications}:
    Across all benchmarks and models, structural specifications reduce code erosion by up to 24.7\%, compared to the task-only prompt. Meanwhile, behavioral specifications increase it. Verbosity varies little across conditions.}
    \label{fig:erosion}
\end{figure}

\begin{figure}[t]
    \centering
    \includegraphics[width=\textwidth]{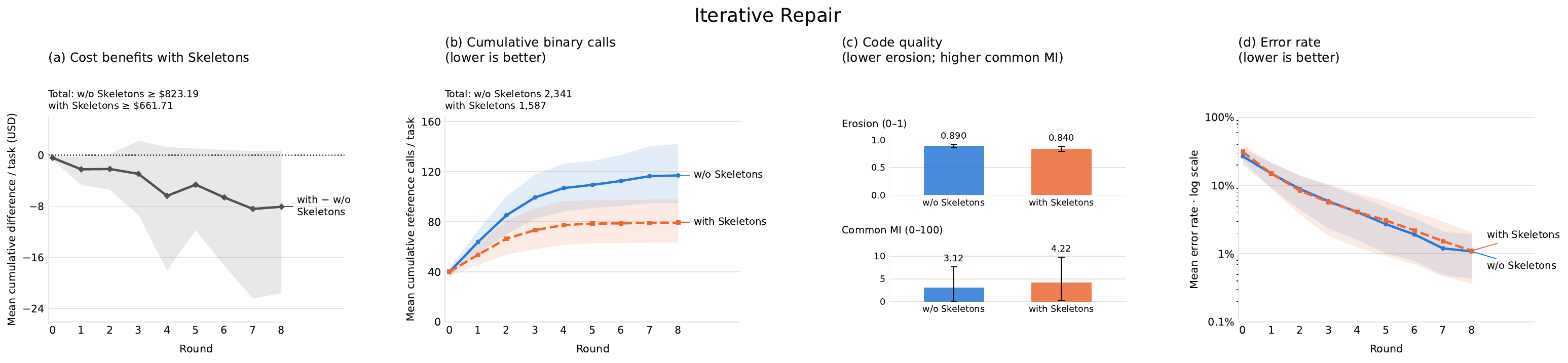}
    \caption{\textbf{Skeletons reduce query effort and cost during iterative repair}:
    For iterative repair tasks, using skeletons reduces the total cost of implementation, and also the number of user queries. We model user queries with the number of calls to the program binary. Additionally, final code has lower erosion and higher maintainability.}
    \label{fig:brownfield_skeleton_coding}
\end{figure}

\paragraph{Structure helps even with perfect behavioral specifications.} One may wonder whether providing enough behavioral specifications can eventually match structural specification. Theoretically, this should be the case, but it does not happen in practice. In \Cref{fig:oracle_tests_scaling}~(Right), we show on ProgramBench that, even when the agent can query the oracle (a full behavioral specification), adding a skeleton with local contracts increases mean pass rate by 10.9\% relative (from 64.0\% to 70.9\%) to the same setup without a skeleton. The best-performing structural condition, a skeleton with global invariants, reduces mean test error across the four models by 27.0\% relative (from 36.0\% to 26.3\%) to that baseline (\Cref{fig:programbench_spec_types}).
To remove the influence of the agent's choice in querying the oracle \Cref{fig:architecture_scaling}~(Right), we provide the increasing number of curated observations up to 160. For any number of curated observations, the performance is strictly worse than that of a skeleton with just 20 observations.

\paragraph{Structural specifications improve smaller models.}
We compare pass rate and implementation cost across model sizes in \Cref{fig:cost_performance}, and test whether structural specifications narrow the gap between smaller and larger models. With detailed global invariants, GPT-5.6 Luna~\cite{openai2026gpt56} improves from 67\% to 90\% on SCBench~\cite{orlanski2026slopcodebench} and from 21.1\% to 79.5\% on SDDBench. These results come within 5.0--8.2 percentage points of the best GPT-5.6 Sol~\cite{openai2026gpt56} results at approximately 15$\times$ and 26$\times$ lower implementation cost, respectively.

\paragraph{More architectural specifications.} Within structural specifications, we have seen in \Cref{fig:pass_rate} that code-like specifications outperform prose-like specifications such as markdown files describing the architecture. This motivates our choice to use code-like specifications for skeleton coding. We can ask if a more elaborate architecture file can bridge the gap with code-like signature specifications. This is not the case: in \Cref{fig:architecture_scaling} we show that adding progressively more architectural information does raise  mean pass rate ($61\% \to 71\%$ on SCBench~\cite{orlanski2026slopcodebench}, $40.1\% \to 52.8\%$ on
SDDBench). However, while it expands the input by roughly 4--7$\times$, every architecture specification still remains below the best code-like specification on the same repositories.

\paragraph{Model fatigue and non-Bayesian models.} 
One may expect that adding more specifications will always result in more faithful implementations. This is, however, not the case. In \Cref{fig:pass_rate}, we take the best-performing skeleton specification, and we add additional local and global contracts that, while correct, are overly detailed. This reduces performance, especially on smaller models. While the Bayesian posterior on additional specifications should only improve pass-rate, prompting the model with specifications need not approximate the Bayesian posterior. This is increasingly the case when more specifications are provided, making actual Bayesian conditioning more complex.

\paragraph{Effect of skeletons on code quality}
An implementation can pass its tests while containing complex or repetitive code. We compare erosion and verbosity across specification types to test whether specifying function boundaries and contracts improves code quality alongside accuracy.
In \Cref{fig:erosion}, we show that signature specifications reduce code erosion from 0.66 to 0.56 on SCBench and from 0.41 to 0.24 on SDDBench relative to the task-only condition, while behavioral specifications increase erosion to as much as 0.73 and 0.43. Verbosity varies less across conditions, but is also lowest under structural specifications on SDDBench.

\paragraph{Skeletons reduce query effort during iterative repair.}
In \Cref{fig:brownfield_skeleton_coding}, we compare iterative fixing of a code-base both without using skeletons and asking the agents to maintain a skeleton. In both cases, we start with the task and reference-binary access; between rounds, we disclose new failing tests and carry the code and, where present, the skeleton into a fresh conversation. By round~8, the skeleton workflow reaches similar error rates with 32\% fewer reference-binary calls (a proxy for behavioral clarification effort), lower mean erosion, and higher mean maintainability.

\section{Related Work}
\label{sec:related}

\paragraph{Contracts and formal specifications.}
Skeletons inherit their basic vocabulary from program logics and
design-by-contract methods, which make interfaces, preconditions,
post-conditions, and invariants explicit rather than leaving them implicit in an
implementation \citep{hoare1969axiomatic,meyer1992designbycontract}. Verification
languages such as Dafny make these contracts machine-checkable and use them to
generate proof obligations for an implementation \citep{leino2010dafny}. The
skeletons we evaluate range from typed stubs \citep{vanrossum2014pep484} to stubs with local contracts and
global invariants written in EARS \citep{mavin2009ears}, and are not machine-checked. We draw on the
contract tradition to separate interfaces, contracts local to one function~\citep{endres2024nl2postcond}, and
invariants shared across components, and measure how much each helps a coding agent.

\paragraph{Intermediate artifacts for code generation.}
Several code-generation systems ask a model to construct an intermediate
artifact before writing or modifying code~\citep{zelikman2023parsel,hong2024metagpt}. Self-planning first elicits a
sequence of implementation steps and then conditions generation on that plan
\citep{jiang2024selfplanning}; CodePlan combines a repository-wide edit plan
with dependency and change-impact analysis \citep{bairi2024codeplan}.
RPG replaces free-form natural-language plans with a repository planning graph
that encodes capabilities, file structure, data flows, and functions, from which
ZeroRepo generates whole repositories \citep{luo2025rpg}.
SpecRover instead infers intended behavior while searching a repository and
uses a reviewer to vet the resulting patch \citep{ruan2025specrover}. Tests can
also serve as intermediate behavioral evidence~\citep{lahiri2022ticoder,fakhoury2024ticoder}: CodeT generates tests to rank
candidate implementations \citep{chen2023codet}. These interventions use
intermediate artifacts, but plans,
specifications, and tests are usually evaluated as complete workflows. We
instead hold the implementation task and model fixed while organizing
skeletons by the contract content they expose, separating that content from
artifact format and raw length.

\paragraph{Specification-driven development.}
Recent work treats specifications as artifacts that guide generation and persist
across revisions. SMART regenerates an ML performance-modeling library from
modular natural-language design documents with worked examples
\citep{kushnir2026designdocs}. He et al.\ compare specification forms and input
granularity, finding benefits from structured pseudocode and module-level inputs
\citep{he2026vibetospec}. SpecFirst separates behavioral specification elicitation
from implementation: one agent combines documentation with observations from querying an execute-only binary into
a specification, which another agent uses for code synthesis
\citep{chen2026specfirst}. We complement these approaches by varying the
structural information supplied in skeletons and measuring reconstruction
accuracy, code quality, and query effort.

\paragraph{Constrained and formally verified generation.}
A complementary line of work enforces stronger constraints during generation~\citep{sun2024clover}.
Program sketching is a direct precedent for skeletons: the programmer writes a
partial program that fixes the structure of a solution, and a synthesizer fills
in the missing details so that the result agrees with a reference
implementation or passes supplied tests \citep{solarlezama2009sketching}.
Synchromesh restricts decoding to programs admitted by a target language and
its semantic constraints \citep{poesia2022synchromesh}. AlphaVerus
\citep{aggarwal2024alphaverus} and AutoVerus \citep{yang2025autoverus} use verifier
feedback, search, and specialized repair procedures to generate verified code and
correctness proofs in Verus, a verifier for Rust \citep{lattuada2023verus}, while the
Vericoding benchmark \citep{bursuc2025vericoding} evaluates generation from formal
specifications across Dafny \citep{leino2010dafny}, Verus, and Lean
\citep{demoura2021lean4}. These systems primarily address the downstream problem of
finding code and proof material for a supplied formal specification.
Our central question is upstream: what information the implementation target
should contain, who must author it, and how much implementation uncertainty it
removes. The current experiments use executable acceptance checks rather than
machine-checked proofs; evaluating the formal checks discussed in
\Cref{sec:iterative_refinement} remains a direction for future experiments.

\paragraph{Greenfield and repository-level evaluation.}
Function-level benchmarks such as HumanEval established test pass rate as a
standard outcome for code generation from a short natural-language description
\citep{chen2021codex}. SlopCodeBench (SCBench) extends evaluation to
long-horizon, iterative greenfield development, where generated decisions
accumulate as a project evolves \citep{orlanski2026slopcodebench}. We study
greenfield and brownfield tasks to expose different burdens: greenfield generation must preserve a
growing design, whereas brownfield modification must first recover relevant
behavior from pre-existing code. Beyond aggregate pass rate, we report outcomes
by specification type, skeleton coverage, model, query effort, and code quality.

\appendix

\section{Specification Types}
\label{sec:specification-types}

All experiments include the natural-language task description and differ in the additional specification provided
to the implementation agent. We consider the following specification types.\footnote{EARS stands for Easy Approach to Requirements
Syntax~\cite{mavin2009ears}; Python signature specifications use stub-style interfaces as described in
PEP~484~\cite{vanrossum2014pep484}; local and global contract specifications build on Design by
Contract~\cite{meyer1992designbycontract}; and architecture specifications use established software-architecture
documentation concepts~\cite{clements2010architectures}.} (i)
\textbf{Task only:} the implementation agent receives only the natural-language task description and no
reference-derived artifact; (ii) \textbf{Concrete Behaviors (EARS):} individual end-to-end behavioral instances
are written as event, condition, and response requirements; (iii) \textbf{Generalized Behaviors (EARS):}
behavioral requirements are expressed as rules over classes of inputs rather than individual examples; (iv)
\textbf{Behavioral Tests:} generated executable tests encode end-to-end input--output behaviors; (v)
\textbf{Design (EARS):} a prose design document describes the system architecture together with EARS behavioral
requirements; (vi) \textbf{Signature Specifications:} the repository structure, modules, types, and callable
signatures are preserved while behavioral contracts and implementation bodies are omitted; (vii) \textbf{Local
Contracts:} signatures are augmented with behavioral contracts attached to individual functions and methods;
(viii) \textbf{Global Invariants:} local contracts are augmented with agreements
that must hold across multiple functions, methods, or modules; (ix)
\textbf{Architecture Specifications:} a prose architecture document progressively adds components and
boundaries, data models and state, collaborations, repository layout, and per-file sketches; and (x)
\textbf{Native Skeleton Specifications:} for ProgramBench~\cite{yang2026programbench}, the source-language file structure, declarations,
types, callable signatures, and EARS contracts are preserved while implementation bodies are omitted.

\subsection{Specification templates}
\label{sec:spec-examples}

To make the specification types concrete, this section shows a template for each artifact, instantiated on a
small synthetic component written for illustration. The component is not drawn from any benchmark repository.
It is a hypothetical library-lending module, \texttt{booklend}, with one class, \texttt{Library}, that tracks
how many copies of each title are on the shelf and on loan. Every template concerns the same design decision: a successful
loan reduces the number of available copies, whereas an availability query leaves it unchanged. Angle brackets
mark the parts that differ between repositories.

\paragraph{(i) Task only.}
A natural-language description of the package, its features, and its public API, written as a user-facing
README. It describes what the library does but not how it is organized.
\needspace{4\baselineskip}
\begin{lstlisting}[language=ears]
# <package name>

## Problem and Purpose
<What the package is for and which problems it solves.>

## Features
- `Library` records the number of copies of each title.
- `lend(title, copies)` lends copies of a title to a reader.
- `available(title)` reports how many copies can still be lent.
\end{lstlisting}

\paragraph{(ii) Concrete Behaviors (EARS).}
A numbered catalogue of end-to-end observations, each a single input and its expected output, written in
EARS syntax (WHEN/IF \ldots\ THEN THE SYSTEM SHALL \ldots). Behavioral scaling experiments draw nested 25\%,
50\%, and 100\% subsets of this catalogue.
\needspace{4\baselineskip}
\begin{lstlisting}[language=ears]
# Behaviors: <package name>

<N> behaviors in <K> capability groups.

1. `result` - WHEN Library holds 5 copies of "poems" and lend("poems", 2) is
   called THEN THE SYSTEM SHALL return True.
2. `state` - WHEN Library holds 5 copies of "poems", lend("poems", 2) is
   called, and available("poems") is called THEN THE SYSTEM SHALL return 3.
3. `state` - WHEN Library holds 5 copies of "poems" and available("poems") is
   called twice THEN THE SYSTEM SHALL return 5 both times.
\end{lstlisting}

\paragraph{(iii) Generalized Behaviors (EARS).}
The same catalogue size, but each entry quantifies over a class of inputs rather than naming one. Note the
WHERE clause and the absence of concrete values.
\needspace{4\baselineskip}
\begin{lstlisting}[language=ears]
1. `state` - WHERE a title has q available copies, WHEN lend is called
   with k <= q copies THEN THE SYSTEM SHALL reduce the available
   copies of that title to q - k.
2. `state` - WHEN available is called for any title THEN THE SYSTEM SHALL
   leave the available copies of every title unchanged.
\end{lstlisting}

\paragraph{(iv) Behavioral Tests.}
Executable pytest functions against the public API only. Each test is the executable form of one concrete
behavior.
\needspace{4\baselineskip}
\begin{lstlisting}[language=Python]
def test_lend_reduces_available():
    library = Library({"poems": 5})
    assert library.lend("poems", 2)
    assert library.available("poems") == 3

def test_available_does_not_change_copies():
    library = Library({"poems": 5})
    assert library.available("poems") == 5
    assert library.available("poems") == 5
\end{lstlisting}

\paragraph{(v) Design (EARS).}
A prose document in the requirements/design format: user stories with EARS acceptance criteria, followed by a
design section describing components, behavioral rules, data models, and error handling in paragraphs.
\needspace{4\baselineskip}
\begin{lstlisting}[language=ears]
### Requirement <n>: <capability>
**User Story:** As a <role>, I want <capability>, so that <benefit>.

#### Acceptance Criteria
1. WHEN a loan succeeds THEN THE SYSTEM SHALL reduce the available
   copies by the number lent.
2. WHEN availability is queried THEN THE SYSTEM SHALL leave all copy counts
   unchanged.

## Design
### Components and Interfaces
- **<component>:** <responsibility, the state it owns, and the components
  it may call.>
\end{lstlisting}

\paragraph{(vi) Signature Specifications.}
The repository as a tree of \texttt{.pyi} stub files: modules, classes, attributes, and callable signatures
with one-paragraph docstrings. Bodies are elided. This is the first level that fixes the internal
representation.
\needspace{4\baselineskip}
\begin{lstlisting}[language=pyi]
class Library:
    """Track available and lent copies per title."""

    _available: dict[str, int]
    _on_loan: dict[str, int]

    def lend(self, title: str, copies: int) -> bool:
        """Lend copies of a title if enough are available."""
        ...

    def available(self, title: str) -> int:
        """Return the number of copies of a title that can still be lent."""
        ...
\end{lstlisting}

\paragraph{(vii) Local Contracts.}
The same stubs, with each docstring extended by argument, return, and exception sections and an
\texttt{Ensures} block of EARS post-conditions scoped to that callable.
\needspace{4\baselineskip}
\begin{lstlisting}[language=pyi]
    def lend(self, title: str, copies: int) -> bool:
        """Lend copies of a title if enough are available.

        Returns:
            True if the loan succeeded, otherwise False.

        Raises:
            KeyError: If ``title`` is unknown.

        Ensures:
            - WHEN the call returns True THEN THE CALLABLE SHALL reduce
              ``_available[title]`` by ``copies``.
            - WHEN the call returns False THEN THE CALLABLE SHALL leave all
              copy counts unchanged.
        """
        ...
\end{lstlisting}

\paragraph{(viii) Global Invariants.}
Local contracts plus a separate \texttt{INVARIANTS.md} file. Each invariant is an EARS statement that spans
several callables, followed by the list of callables that must agree with it; the invariant is also cited
from the docstrings of those callables.
\needspace{4\baselineskip}
\begin{lstlisting}[language=ears]
## I<n> -- Lent and available copies stay consistent
WHEN any public method returns THEN THE SYSTEM SHALL keep, for every title,
the sum of its available and lent copies equal to the number of copies it owns.

Participants: Library.lend, Library.checkin, Library.available
\end{lstlisting}

\needspace{12\baselineskip}
\paragraph{(ix) Architecture Specifications.}
We illustrate the five architecture variants in \Cref{fig:architecture_scaling} using \texttt{booklend}.
All variants accompany the task description. The first specifies component responsibilities and state
ownership. Each subsequent variant retains the preceding sections and adds a different kind of information:
the state representation and its invariants, the sequence of checks and updates, the repository layout,
and finally the interfaces and requirements of individual files. Thus, the first variant includes only
section~1 below, and the final variant includes all five sections.
\needspace{4\baselineskip}
\begin{lstlisting}[language=ears]
## 1. Components and Boundaries
- Library manages an in-memory inventory and exposes lend and available.
- Library owns all mutable lending state. Callers use its public methods
  rather than modifying copy counts directly.
- Lending may change state; availability queries only read it.

## 2. Data Models and State
- _available maps each title to its number of available copies.
- _on_loan maps each title to its number of lent copies.
- Both mappings contain the same known titles and nonnegative counts.
- Initially, _available contains the supplied inventory and _on_loan
  contains zero for every title.
- For each title, available plus lent copies equals its initial inventory.

## 3. Collaborations
- A caller submits a title and copy count to Library.lend.
- lend validates the title and positive copy count before changing state.
- If stock is insufficient, lend returns False without changing state.
- Otherwise, lend decreases _available and increases _on_loan by the
  requested count, then returns True.
- available validates the title and reads _available; it changes no state.

## 4. Repository Layout
- booklend/__init__.py: exports the public Library class.
- booklend/library.py: defines Library, initializes its state, and
  implements lending and availability queries.
- The package entry point imports Library from library.py; library.py
  does not depend on the package entry point.

## 5. File Sketches
- booklend/__init__.py re-exports Library without creating an inventory.
- booklend/library.py defines Library(inventory), lend(title, copies),
  and available(title), with the two state mappings described above.
- WHEN lend returns True THEN THE CALLABLE SHALL reduce _available[title]
  and increase _on_loan[title] by copies before returning.
- WHEN lend returns False THEN THE CALLABLE SHALL leave all counts unchanged.
- WHEN available returns THEN THE CALLABLE SHALL return _available[title]
  and leave both mappings unchanged.
- IF a title is unknown THEN THE CALLABLE SHALL raise KeyError before
  changing either mapping.
\end{lstlisting}

\paragraph{(x) Native Skeleton Specifications (ProgramBench).}
For ProgramBench tasks, the skeleton is written in the source language of the task. It keeps the module tree,
type declarations, trait or interface implementations, and function signatures with their documentation
comments, and replaces every body with a placeholder. Module-level EARS contracts are placed in a
specification file beside each directory. The template below is written in Rust.
\needspace{4\baselineskip}
\begin{lstlisting}[language=rustmini]
/// <What the type represents and the invariants it maintains.>
pub struct Library {
    available: HashMap<String, u32>,
}

impl Library {
    /// <What the function does, its arguments, and its result.>
    pub fn lend(&mut self, title: &str, copies: u32) -> bool {
        unimplemented!("specification stub")
    }
}
\end{lstlisting}

\section{SDDBench}
\label{app:sddbench}

SDDBench is a repository-level benchmark for evaluating
specification-driven software development with coding agents. We derive its tasks from
real Python libraries, with renamed packages and rephrased project
descriptions to reduce direct recognition of the original projects. In our
experiments, the agent reconstructs each library from a self-contained task
description and the specification supplied by the experimental condition,
without access to the reference implementation. Functional correctness is the
fraction of tests passed in the corresponding project's evaluation suite, which
is executed separately from the implementation. This setup isolates how the
structure and behavioral detail of a specification affect reconstruction
accuracy, implementation cost, and code quality. Since the underlying
projects are public, anonymization reduces explicit identity cues but does not completely
rule out training-data contamination.

\begin{center}
\small
\begin{tabular}{ll}
\hline
\textbf{Real name} & \textbf{Anonymized name} \\
\hline
Babel & LocaleFoundry \\
CacheControl & RequestsCachePro \\
cachetools & memokit \\
Click & CLIForge \\
icecream & debugwhisper \\
itsdangerous & cryptoseal \\
marshmallow & serializor \\
Mimesis & MirageMint \\
Parsel & DataHarvest \\
PySnooper & CodeTrail \\
statsmodels & StatWeaver \\
TinyDB & MicroVault \\
xarray & AxisAtlas \\
\hline
\end{tabular}
\end{center}

\needspace{9\baselineskip}
\section{ProgramBench and SCBench}
\label{app:programbench}

We evaluated 20 tasks from ProgramBench~\cite{yang2026programbench}:
\texttt{atlas}, \texttt{brotli}, \texttt{dropbear}, \texttt{dupl}, \texttt{fzf},
\texttt{go-mod-outdated}, \texttt{loop}, \texttt{luajit}, \texttt{ninja},
\texttt{nomino}, \texttt{revive}, \texttt{richgo}, \texttt{ripsecrets},
\texttt{rust-sloth}, \texttt{sd}, \texttt{solar}, \texttt{tty-clock},
\texttt{typst}, \texttt{xq}, and \texttt{xz}.
For our main SCBench~\cite{orlanski2026slopcodebench} experiments, we evaluated
reconstructions of each problem's final repository state.

\section{Experimental Details}
\label{app:experimental_details}

Here we describe the details of the experiments in \Cref{sec:results}.
We ran all experiments with the AI Functions
framework (\url{https://github.com/strands-labs/ai-functions}).
We used the models, specification formats, and metrics described in
\Cref{sec:experiments,sec:specification-types}.
In \Cref{fig:oracle_tests_scaling_a}, we evaluated four models on 20
ProgramBench tasks~\cite{yang2026programbench}~(\Cref{app:programbench}).
For each task, we held out approximately 25\% of the eligible reference-passing
tests for evaluation and used the remaining tests as the available pool.
We gave each agent the task, reference-executable access, and nested
0/33.3/66.7/100\% subsets of that pool, drawn with a fixed seed.
We scored all conditions on the same held-out set. Thus, 100\% on the x-axis
means that all tests in the available pool were provided, while the evaluation
tests remained withheld.
In \Cref{fig:programbench_spec_types}, we used the same tasks and models
to compare task-only input with prose in EARS
format~\cite{mavin2009ears}, original comments, signatures, local contracts,
or global invariants. We retained reference-executable access in every
condition and used reference source only for artifact generation.
In \Cref{fig:pass_rate}, we evaluated four models on 35 SCBench
problems~\cite{orlanski2026slopcodebench} and 13 SDDBench libraries
(\Cref{app:sddbench}). We supplied the task alone or one additional artifact,
using full catalogues for behavioral conditions and withholding reference
code and evaluation tests.
In \Cref{fig:erosion}, we measured erosion and verbosity on the same
reconstructions using the SCBench metrics~\cite{orlanski2026slopcodebench}.

\begin{figure}[t]
    \centering
    \begin{minipage}[t]{0.32\textwidth}
        \vspace{0pt}
        \includegraphics[width=\linewidth,trim=0 0 446 45,clip]{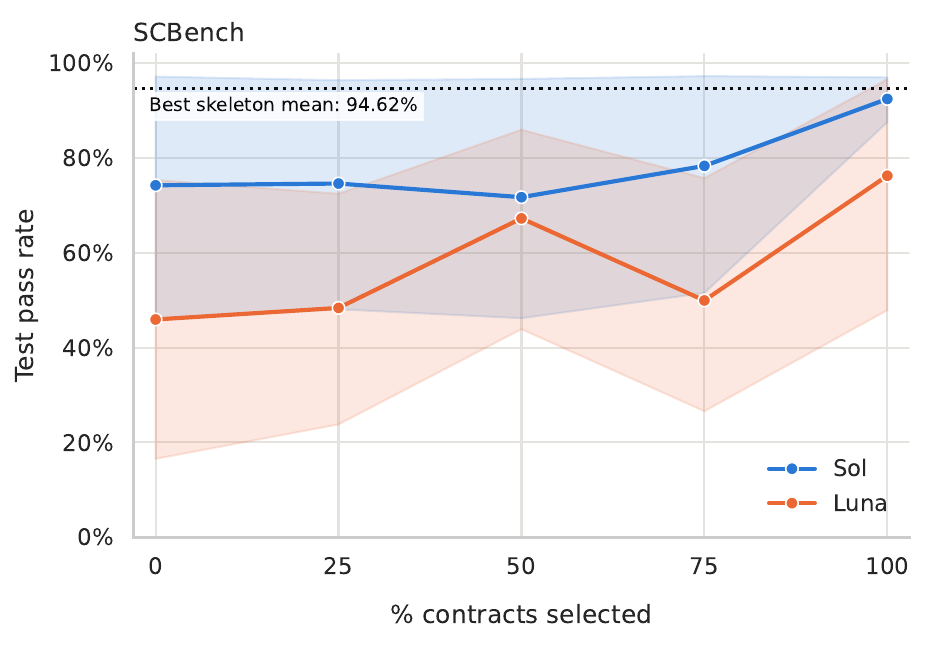}
    \end{minipage}\hfill
    \begin{minipage}[t]{0.32\textwidth}
        \vspace{0pt}
        \includegraphics[width=\linewidth,trim=445 0 0 45,clip]{figures/figure09ab.pdf}
    \end{minipage}\hfill
    \begin{minipage}[t]{0.32\textwidth}
        \vspace{0pt}
        \includegraphics[width=\linewidth,height=0.72\linewidth]{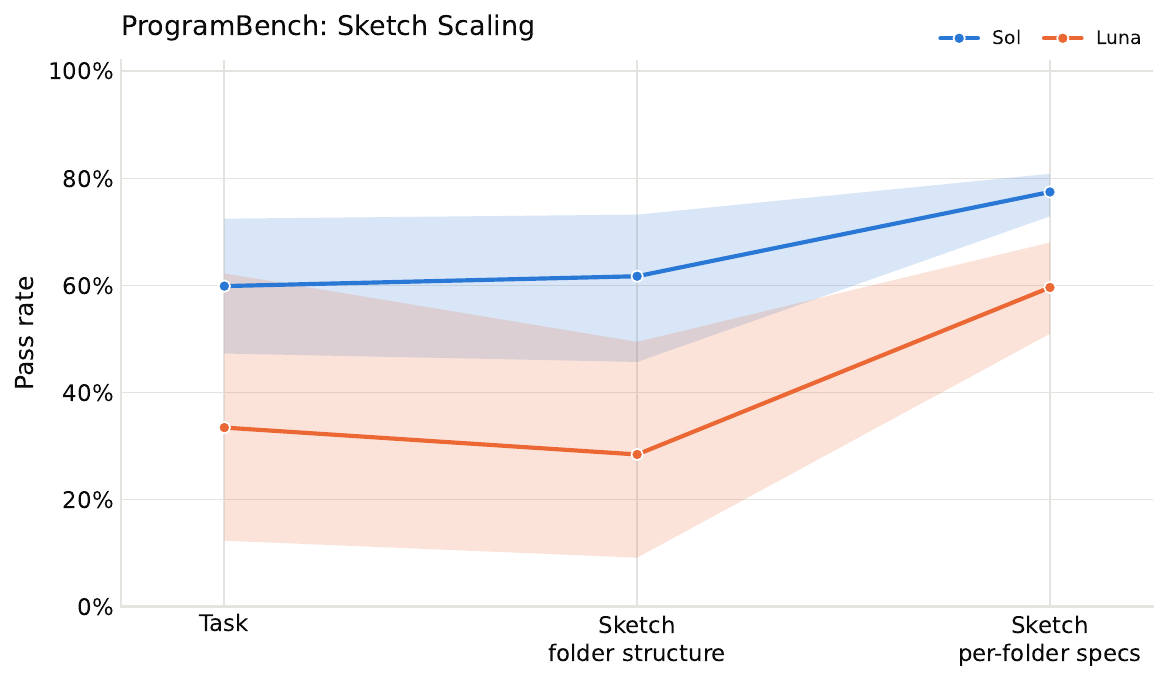}
    \end{minipage}
    \caption{\textbf{Skeleton coverage vs. pass rates}:
    \textbf{(Left, Middle)} We select a subset of the functions to add to the skeleton based on their complexity. Providing additional coverage increases pass-rate. \textbf{(Right)} Adding per-folder descriptions to a directory sketch improves the pass rates.}
    \label{fig:complexity_scaling}
    \label{fig:programbench_sketch_scaling}
\end{figure}

In \Cref{fig:architecture_scaling}~(Right), we allowed the four models 20, 40,
80, or 160 reference-executable queries on ProgramBench, with or without a
Global Invariants skeleton. We also supplied curated examples at matching
counts or allowed uncapped querying.
In \Cref{fig:architecture_scaling}, we used five repositories from
SCBench and four from SDDBench with four models. We reduced each architecture document to five nested
versions covering components, state, collaborations, layout, and per-file
sketches, and compared them with task-only input. We counted both task and
architecture text in the input size.
In \Cref{fig:complexity_scaling} (left, middle), we evaluated Sol and Luna on
five repositories per benchmark. We ranked callables by complexity, retained
nested 25/50/75\% subsets of their local contracts, and compared these with
task-only input and the full skeleton. We counted coverage by callable and
scored reconstructions with benchmark tests.
In \Cref{fig:cost_performance}, we compared accuracy and implementation cost
for the models and conditions in \Cref{fig:pass_rate}. We priced input, output, cache-read, and
cache-write tokens at fixed rates, excluding specification generation and
evaluation compute.
In \Cref{fig:brownfield_skeleton_coding}, we compared code-only and
agent-maintained specification workflows on 20 ProgramBench task pairs,
using Sol. We started both from the task and
reference-executable access, disclosed new failing tests between rounds,
and carried code and specifications into fresh conversations. We tracked
error, model cost, and reference calls through round~8 and measured code
quality at that point. For common maintainability, we combined source length,
cyclomatic complexity, and Halstead volume.

\begin{figure}[htbp]
    \centering
    \includegraphics[width=0.6\textwidth]{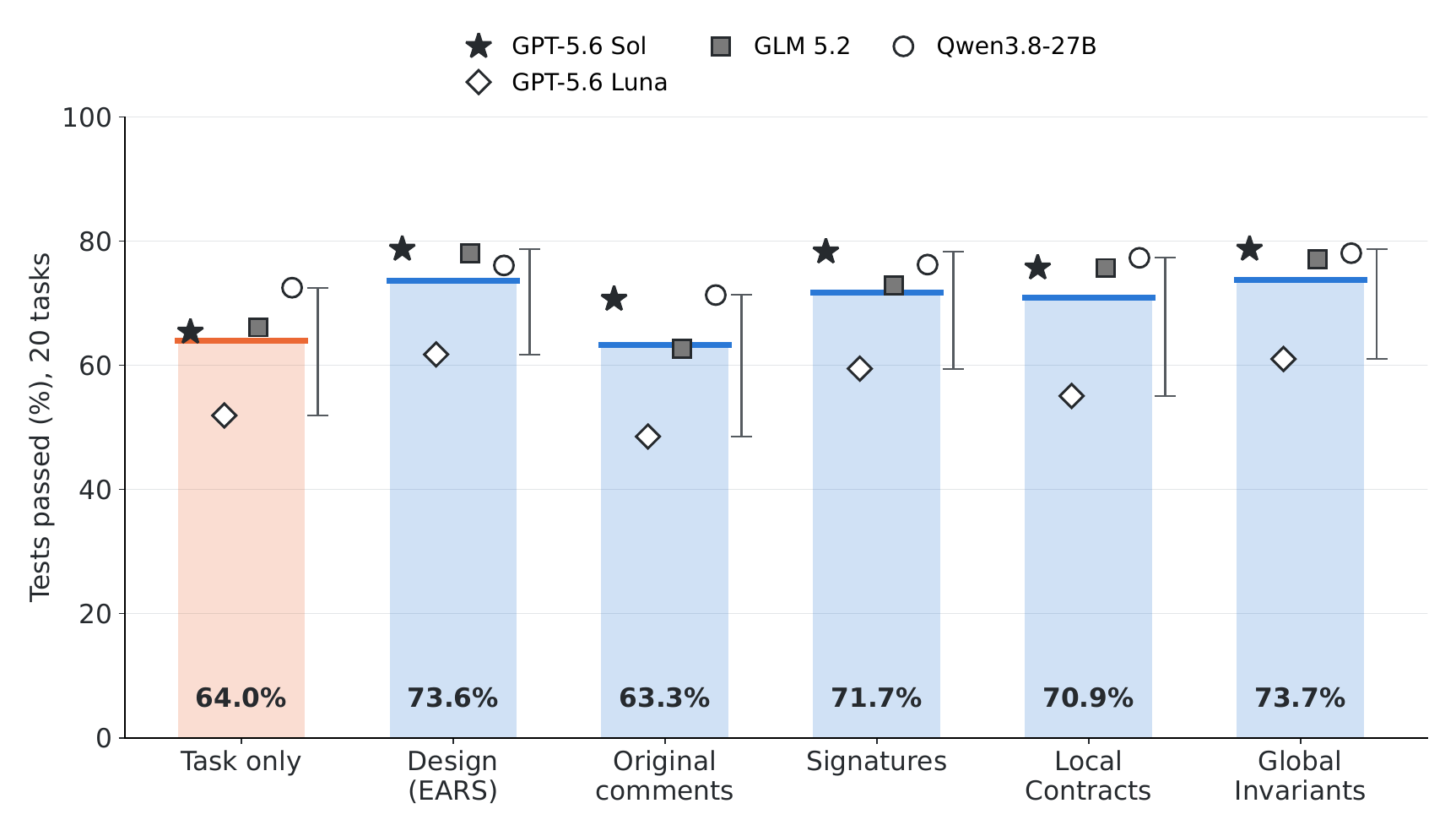}
    \caption{\textbf{Value of different types of skeletons}: Providing skeletons improves the pass-rates on ProgramBench despite query access to the oracle executable.}
    \label{fig:programbench_spec_types}
\end{figure}

\section{Proofs and remarks for \Cref{sec:method,sec:iterative_refinement}}
\label{app:proofs_method}

Throughout, for a fixed skeleton $S$, $G=\mathcal G\cap\mathcal H_S$ denotes
the acceptable programs with skeleton $S$, and $V_D$ the programs in
$\mathcal H_S$ that agree with the reference on every input tested in $D$.
Exact conditioning on $D$ gives
\[
    p_S(\mathcal G\mid D)=p_S(G\mid D)=\frac{p_S(G\cap V_D)}{p_S(V_D)}.
\]

Also, the sum in \cref{eq:goal_mixture} is over canonical skeletons defined by a fixed map
$\sigma:\mathcal H\to\mathcal S$, with
$\mathcal H_S=\{h\in\mathcal H:\sigma(h)=S\}$.
These classes partition $\mathcal H$. A partial skeleton is treated as
conditioning on the union of compatible canonical classes. Throughout,
we condition only on events of positive probability, and omit
zero-weight terms from the sum.

\subsection{Proof of \texorpdfstring{Theorem~\ref{thm:structural_end_to_end_separation}}{Theorem structural_end_to_end_separation}}
\label{app:proof_skeleton_separation}
\begin{proof}
Take inputs $x\in\{0,1\}^d$ and the parity target
$f^\ast(x)=x_1\oplus\cdots\oplus x_d$. A program is acceptable exactly when it computes
$f^\ast$ on every input. For $a\in\{0,1\}^d$, the two
skeletons permit respectively
\[
\begin{aligned}
    h_a^A(x)&=f^\ast(x)\oplus\mathbf 1\{a\ne0,\ x=a\},\\
    h_a^B(x)&=\bigoplus_{i=1}^d(1-a_i)x_i
             =f^\ast(x)\oplus\langle a,x\rangle,
\end{aligned}
\]
where the inner product is over $\mathbb F_2$. The first structure allows a
joint exceptional branch; the second combines contributions by XOR. Each
parameter setting has prior mass $2^{-d}$ under each skeleton, the $2^d$
programs of each skeleton have distinct behaviors, and both skeletons
implement the target exactly when $a=0$.

In $S_A$, each incorrect completion differs from the target only at $x=a$,
so all $2^d-1$ nonzero inputs must be tested to eliminate them. In $S_B$,
each test imposes $\langle a,x\rangle=0$. Inputs of rank $r$ leave
$2^{d-r}$ candidates, so $d$ tests are necessary and the standard basis
inputs suffice. Under either uniform prior, any suite leaving at least
two candidates has posterior correctness at most $1/2<1-\delta$, while a
suite leaving only the target has posterior correctness one. Hence
$N_\delta(S_A)=2^d-1$ and $N_\delta(S_B)=d$.

For the second part, let $D$ consist of the $d$ tests at the standard basis
vectors $e_1,\ldots,e_d$, which guarantee success under $S_B$. Under $S_A$,
the test at $e_i$ eliminates only $a=e_i$, so $2^d-d$ programs survive, one of
them correct, and the success probability is $1/(2^d-d)$. Now let $p_0$ put
mass $1/2$ on each skeleton, uniformly within each. Under $S_B$ only $a=0$
survives, with prior mass $2^{-(d+1)}$; under $S_A$ the survivors have total
mass $(2^d-d)\,2^{-(d+1)}$. Both surviving programs with $a=0$ are
acceptable, so
\[
    p_0(S_A\mid D)=\frac{2^d-d}{2^d-d+1},\qquad
    p_0(\mathcal G\mid D)=\frac{2\cdot2^{-(d+1)}}{(2^d-d+1)\,2^{-(d+1)}}
                        =\frac{2}{2^d-d+1}.
\]
\end{proof}

The information supplied by the same labeled sample depends on the
skeleton. Before any tests, fix a nonzero
input $x$ and supply $D_x=\{(x,f^\ast(x))\}$. Measuring information gain
in bits by the posterior-to-prior KL divergence gives
\begin{equation}\label{eq:same_sample_information}
\begin{aligned}
    \mathrm{IG}_{S_A}(D_x)
    &=\KL\!\left(p_{S_A}(\cdot\mid D_x)\,\|\,p_{S_A}\right)
      =\log_2\frac{2^d}{2^d-1},\\
    \mathrm{IG}_{S_B}(D_x)
    &=\KL\!\left(p_{S_B}(\cdot\mid D_x)\,\|\,p_{S_B}\right)
      =1.
\end{aligned}
\end{equation}
Indeed, the posteriors are uniform on $2^d-1$ and $2^{d-1}$ survivors,
respectively. Thus the same reference-labeled sample supplies only
$\Theta(2^{-d})$ bits under $S_A$, but one bit under $S_B$.
For comparison, under each skeleton let the unknown parameter $A$ be
uniform and generate the response $Y=h_A^S(x)$. The expected information
gains are
\[
    I(A;Y\mid S=S_A)=h_2(2^{-d}),\qquad
    I(A;Y\mid S=S_B)=1,
\]
where $x$ remains fixed and
$h_2(u)=-u\log_2u-(1-u)\log_2(1-u)$ is binary entropy.
The first response differs from $f^\ast(x)$ with probability $2^{-d}$;
the second splits the prior equally between the two possible outputs.
These mutual informations average over possible responses, whereas
\Cref{eq:same_sample_information} concerns the same observed response.

\subsection{Proof of \texorpdfstring{Theorem~\ref{thm:blue_local_separation}}{Theorem blue_local_separation}}
\label{app:proof_local_separation}
\begin{proof}
For each $a\in\{0,1\}^n$, choose
\[
    h_1(x,z)=\mathbf 1\{x=a_1\},\qquad
    h_i(x,z)=z\land\mathbf 1\{x=a_i\}\quad(2\le i\le n).
\]
This implements $F_a(x)=\mathbf 1\{x=a\}$. The same architecture also
implements $F_\emptyset(x)\equiv0$, for example with every transition
constant zero. Thus this small compositional class contains the
singleton-indicator family~\cite{liu2016teaching}.

Put prior mass $1/2$ on the constant-zero implementation, which serves as
the reference, and mass $2^{-(n+1)}$ on the implementation above of each
$F_a$. Let $G=\{h\in\mathcal H_S:F_h\equiv0\}$, so that $p_S(G)=1/2$. Every
test is labeled $0$, so adaptive selection is no more informative than a
fixed suite, and a test at input $x$ eliminates exactly the implementation
of $F_x$. For $m$ distinct tests,
\[
    p_S(G\mid D)=\frac{1/2}{1-m/2^{n+1}}
               =\frac{1}{2-m/2^n}.
\]
Repeated inputs provide no additional constraint. Solving
$p_S(G\mid D)\ge1-\delta$ gives the lower bound on $N_\delta(S)$. For
$\delta=0$ the bound equals $2^n$, and testing every input suffices. Locally,
four input pairs reveal the truth table of each transition under any prior,
so $4n$ local tests identify every transition table.
\end{proof}

\paragraph{Worst case without a prior.}
When all Boolean transition tables are allowed, the same construction shows
that exact identification of $F_h$ requires $2^n$ end-to-end queries in the
worst case, even with adaptive input selection: on a transcript consisting
entirely of zero outputs, every unqueried input $a$ leaves $F_a$
indistinguishable from $F_\emptyset$. Testing every input suffices for any
composite behavior.

The system has only $4n$ unknown truth-table bits, but its endpoint can
hide where a disagreement arose. In the construction, the state records
whether every input bit seen so far matches $a$. After a mismatch, later
stages keep the state at zero. An endpoint test then reveals neither which
stage rejected the input nor what the other stages would do. A local test
observes a component before subsequent transitions can erase its effect.
Software testability theory describes this distinction through fault
execution, infection of internal state, and propagation to an observable
output~\cite{voas1992pie}. Direct access to inputs and outputs improves a
component's controllability and observability~\cite{freedman1991testability}.

State merging permits this example; depth alone does not imply hardness.
For comparison, if $h_i(x,0)\ne h_i(x,1)$ for both values of $x$, every
transition has the form $h_i(x,z)=z\oplus u_i x\oplus v_i$.
The composition is then affine over bits and is identified by the $n+1$
endpoint queries at $0^n$ and the standard basis vectors. The exponential
result concerns exact behavioral identification and a specified prior:
it does not establish a bound for every prior or for prediction at fixed
nonzero error. Endpoint data need not identify the individual transition
tables uniquely, even when they identify the composite behavior.

\subsection{Proof of \texorpdfstring{Theorem~\ref{thm:blue_cost_decomposition}}{Theorem blue_cost_decomposition}}
\label{app:proof_cost_decomposition}
\begin{proof}
Here $\mathcal N_0(S)$ is the minimum total number of local and end-to-end
tests needed for posterior correctness one. All tests are answered by the
fixed acceptable reference $h^\ast$. Each minimum test count is $+\infty$
when no finite suite suffices.
For any classes $G_i$ containing the reference components and having finite
costs, combine local suites excluding their complements with a joint suite
excluding the unacceptable combinations that remain. All survivors belong
to $G$, and the reference survives with positive mass. This gives the upper
bound.

Conversely, take any successful mixed suite, and let $G_i$ be the component
implementations passing its local tests in slot $i$. The number of those
tests is at least $N_i(G_i)$. Its end-to-end part must exclude every
unacceptable member of $\mathcal H_S\cap\prod_i G_i$, since each such member
has positive prior mass and passes all local tests. That part therefore
costs at least $J_S(G_1,\ldots,G_k)$. Minimizing proves equality, including
the case in which no finite suite succeeds. The zero-cost characterization
follows because, without joint tests, all of these combinations survive.
\end{proof}

\subsection{Proof of \texorpdfstring{Proposition~\ref{prop:blue_behavioral_obstruction}}{Proposition blue_behavioral_obstruction}}
\label{app:proof_behavioral_obstruction}
\begin{proof}
The reference $h_0\in\mathcal H_S$ passes every test in $D$, and we assume
$p_S([h_0])>0$. Every member of $[h_0]$ also passes every test, so
conditioning on $D$ leaves the relative weights within $[h_0]$ unchanged.
Since $G\subseteq[h_0]\subseteq V_D$, exact conditioning gives
\[
    p_S(\mathcal G\mid D)=\frac{p_S(G)}{p_S(V_D)}
    \le \frac{p_S(G)}{p_S([h_0])}.
\]
If $1-\delta$ exceeds this ratio, no finite test suite attains the required
posterior correctness, so $N_\delta(S)=\infty$.
\end{proof}

\subsection{Proof of \texorpdfstring{Proposition~\ref{prop:component_rejection}}{Proposition component_rejection}}
\label{app:proof_component_rejection}
\begin{proof}
Each retry for component $i$ is an independent draw from the same
distribution $q_i$, with $\alpha_i>0$. Each check terminates and accepts
exactly the implementations in $G_i$.
By independence, a whole program passes every check with probability
$\prod_i\alpha_i$, so the number of proposals is geometric with mean
$\prod_i1/\alpha_i$, and an accepted program is distributed as
$q(\cdot\mid G)$. Resampling component $i$ until it passes takes a geometric
number of proposals with mean $1/\alpha_i$ and returns a sample of
$q_i(\cdot\mid G_i)$. By independence, the accepted components are jointly
distributed as $\prod_i q_i(\cdot\mid G_i)=q(\cdot\mid G)$.
\end{proof}

\subsection{Proof of \texorpdfstring{Proposition~\ref{prop:judge}}{Proposition judge}}
\label{app:proof_judge}
\begin{proof}
We retain the sampling assumptions of Proposition~\ref{prop:component_rejection},
with the checks specified in this proposition. All error rates are measured
under the fixed proposal distributions; in particular, when $\alpha_i<1$,
$\varepsilon_i=\Pr(\mathrm{accept}\mid h_i\notin G_i)$.
The checker uses fresh randomness, independent across attempts and components.
A violation of contract $i$ either is consistent with the docstring, with
probability $\beta_i$, and is then accepted, or is not and is accepted at rate
$e_i$; this gives $\varepsilon_i=\beta_i+(1-\beta_i)e_i$.

(i) Every implementation that satisfies the contract is accepted, so by
Bayes' rule an accepted component satisfies its contract with probability
$\alpha_i/(\alpha_i+(1-\alpha_i)\varepsilon_i)$.
If $\alpha_i=1$ or $\varepsilon_i=0$, this probability is one.

(ii) The accepted components are independent, so by (i) they all satisfy
their contracts with probability $\prod_i1/(1+x_i)$, where
$x_i=\varepsilon_i(1-\alpha_i)/\alpha_i$. Since $1/(1+x)\ge\max\{0,1-x\}$ for
$x\ge0$ and $\prod_i\max\{0,1-x_i\}\ge1-\sum_ix_i$, the lower bound follows.

For the whole-program comparison in (ii), write
$\alpha=q(G)=\prod_i\alpha_i$. If $\alpha=1$, correctness is certain.
Otherwise, let $\varepsilon=\Pr(\mathrm{accept}\mid h\notin G)>0$.
As in (i), the accepted program satisfies every contract with probability
\[
    \frac{\alpha}{\alpha+(1-\alpha)\varepsilon}
    \le \frac{\alpha}{\varepsilon},
\]
since $\alpha+(1-\alpha)\varepsilon\ge\varepsilon$.
When $\varepsilon=0$, every accepted program satisfies all contracts.
\end{proof}


\begin{thebibliography}{51}
\providecommand{\natexlab}[1]{#1}
\providecommand{\url}[1]{\texttt{#1}}
\expandafter\ifx\csname urlstyle\endcsname\relax
  \providecommand{\doi}[1]{doi: #1}\else
  \providecommand{\doi}{doi: \begingroup \urlstyle{rm}\Url}\fi

\bibitem[Dijkstra(1982)]{dijkstra1982role}
Edsger~W. Dijkstra.
\newblock {On the Role of Scientific Thought}.
\newblock In \emph{Selected Writings on Computing: A Personal Perspective},
  pages 60--66. Springer, 1982.

\bibitem[{Amazon Web Services}(2025)]{kiro2025}
{Amazon Web Services}.
\newblock {Kiro: An Agentic IDE for Spec-Driven Development}.
\newblock \url{https://kiro.dev}, 2025.

\bibitem[van~der Walt et~al.(2011)van~der Walt, Colbert, and
  Varoquaux]{vanderwalt2011numpy}
St{\'e}fan van~der Walt, S.~Chris Colbert, and Ga{\"e}l Varoquaux.
\newblock {The NumPy Array: A Structure for Efficient Numerical Computation}.
\newblock \emph{Computing in Science \& Engineering}, 13\penalty0 (2):\penalty0
  22--30, 2011.

\bibitem[Harris et~al.(2020)Harris, Millman, van~der Walt,
  et~al.]{harris2020numpy}
Charles~R. Harris, K.~Jarrod Millman, St{\'e}fan~J. van~der Walt, et~al.
\newblock {Array Programming with NumPy}.
\newblock \emph{Nature}, 585:\penalty0 357--362, 2020.

\bibitem[{\AA}str{\"o}m and Eykhoff(1971)]{astrom1971identification}
Karl~Johan {\AA}str{\"o}m and Pieter Eykhoff.
\newblock {System Identification---A Survey}.
\newblock \emph{Automatica}, 7\penalty0 (2):\penalty0 123--162, 1971.

\bibitem[Gold(1967)]{gold1967identification}
E.~Mark Gold.
\newblock {Language Identification in the Limit}.
\newblock \emph{Information and Control}, 10\penalty0 (5):\penalty0 447--474,
  1967.

\bibitem[Liu and Zhu(2016)]{liu2016teaching}
Ji~Liu and Xiaojin Zhu.
\newblock {The Teaching Dimension of Linear Learners}.
\newblock \emph{Journal of Machine Learning Research}, 17\penalty0
  (175):\penalty0 1--25, 2016.
\newblock \url{https://jmlr.org/papers/v17/15-630.html}.

\bibitem[Falck et~al.(2024)Falck, Wang, and Holmes]{falck2024martingale}
Fabian Falck, Ziyu Wang, and Chris Holmes.
\newblock {Is In-Context Learning in Large Language Models Bayesian? A
  Martingale Perspective}.
\newblock In \emph{Proceedings of the 41st International Conference on Machine
  Learning}, 2024.

\bibitem[Leino(2010)]{leino2010dafny}
K.~Rustan~M. Leino.
\newblock {Dafny: An Automatic Program Verifier for Functional Correctness}.
\newblock In \emph{Logic for Programming, Artificial Intelligence, and
  Reasoning}, volume 6355 of \emph{Lecture Notes in Computer Science}, pages
  348--370. Springer, 2010.
\newblock \url{https://doi.org/10.1007/978-3-642-17511-4_20}.

\bibitem[Lattuada et~al.(2023)Lattuada, Hance, Cho, Brun, Subasinghe, Zhou,
  Howell, Parno, and Hawblitzel]{lattuada2023verus}
Andrea Lattuada, Travis Hance, Chanhee Cho, Matthias Brun, Isitha Subasinghe,
  Yi~Zhou, Jon Howell, Bryan Parno, and Chris Hawblitzel.
\newblock {Verus: Verifying Rust Programs using Linear Ghost Types}.
\newblock \emph{Proceedings of the ACM on Programming Languages}, 7\penalty0
  (OOPSLA1):\penalty0 286--315, 2023.
\newblock \doi{10.1145/3586037}.

\bibitem[de~Moura and Ullrich(2021)]{demoura2021lean4}
Leonardo de~Moura and Sebastian Ullrich.
\newblock {The Lean 4 Theorem Prover and Programming Language}.
\newblock In \emph{Automated Deduction -- CADE 28}, volume 12699 of
  \emph{Lecture Notes in Computer Science}, pages 625--635. Springer, 2021.
\newblock \doi{10.1007/978-3-030-79876-5_37}.

\bibitem[Orlanski et~al.(2026)Orlanski, Roy, Yun, Shin, Gu, Ge, Adila, Roberts,
  Sala, and Albarghouthi]{orlanski2026slopcodebench}
Gabriel Orlanski, Devjeet Roy, Alexander Yun, Changho Shin, Alex Gu, Albert Ge,
  Dyah Adila, Nicholas Roberts, Frederic Sala, and Aws Albarghouthi.
\newblock {SlopCodeBench: Benchmarking How Coding Agents Degrade Over
  Long-Horizon Iterative Tasks}.
\newblock \emph{arXiv preprint arXiv:2603.24755}, 2026.
\newblock \url{https://arxiv.org/abs/2603.24755}.

\bibitem[Yang et~al.(2026)Yang, Lieret, Ma, Thakkar, Pedchenko, Sootla,
  McMilin, Yin, Hou, Synnaeve, Yang, and Press]{yang2026programbench}
John Yang, Kilian Lieret, Jeffrey Ma, Parth Thakkar, Dmitrii Pedchenko, Sten
  Sootla, Emily McMilin, Pengcheng Yin, Rui Hou, Gabriel Synnaeve, Diyi Yang,
  and Ofir Press.
\newblock {ProgramBench: Can Language Models Rebuild Programs From Scratch?}
\newblock \emph{arXiv preprint arXiv:2605.03546}, 2026.
\newblock \url{https://arxiv.org/abs/2605.03546}.

\bibitem[{OpenAI}(2026)]{openai2026gpt56}
{OpenAI}.
\newblock {GPT-5.6: Frontier Intelligence for Professional Work}.
\newblock \url{https://openai.com/index/gpt-5-6/}, 2026.

\bibitem[Claessen and Hughes(2000)]{claessen2000quickcheck}
Koen Claessen and John Hughes.
\newblock {QuickCheck: A Lightweight Tool for Random Testing of Haskell
  Programs}.
\newblock In \emph{Proceedings of the Fifth ACM SIGPLAN International
  Conference on Functional Programming}, pages 268--279, 2000.

\bibitem[Meyer(1992)]{meyer1992designbycontract}
Bertrand Meyer.
\newblock {Applying ``Design by Contract''}.
\newblock \emph{Computer}, 25\penalty0 (10):\penalty0 40--51, 1992.
\newblock \url{https://doi.org/10.1109/2.161279}.

\bibitem[Goldman and Kearns(1995)]{goldman1995teaching}
Sally~A. Goldman and Michael~J. Kearns.
\newblock {On the Complexity of Teaching}.
\newblock \emph{Journal of Computer and System Sciences}, 50\penalty0
  (1):\penalty0 20--31, 1995.

\bibitem[Angluin(1988)]{angluin1988queries}
Dana Angluin.
\newblock {Queries and Concept Learning}.
\newblock \emph{Machine Learning}, 2\penalty0 (4):\penalty0 319--342, 1988.

\bibitem[Kalman(1963)]{kalman1963mathematical}
Rudolf~E. Kalman.
\newblock {Mathematical Description of Linear Dynamical Systems}.
\newblock \emph{Journal of the Society for Industrial and Applied Mathematics,
  Series A: Control}, 1\penalty0 (2):\penalty0 152--192, 1963.

\bibitem[Voas(1992)]{voas1992pie}
Jeffrey~M. Voas.
\newblock {PIE: A Dynamic Failure-Based Technique}.
\newblock \emph{IEEE Transactions on Software Engineering}, 18\penalty0
  (8):\penalty0 717--727, 1992.
\newblock \url{https://doi.org/10.1109/32.153381}.

\bibitem[Freedman(1991)]{freedman1991testability}
Roy~S. Freedman.
\newblock {Testability of Software Components}.
\newblock \emph{IEEE Transactions on Software Engineering}, 17\penalty0
  (6):\penalty0 553--564, 1991.
\newblock \url{https://doi.org/10.1109/32.87281}.

\bibitem[Parnas(1972)]{parnas1972criteria}
David~L. Parnas.
\newblock {On the Criteria To Be Used in Decomposing Systems into Modules}.
\newblock \emph{Communications of the ACM}, 15\penalty0 (12):\penalty0
  1053--1058, 1972.
\newblock \url{https://doi.org/10.1145/361598.361623}.

\bibitem[Benveniste et~al.(2018)Benveniste, Caillaud, Nickovic, Passerone,
  Raclet, Reinkemeier, Sangiovanni-Vincentelli, Damm, Henzinger, and
  Larsen]{benveniste2018contracts}
Albert Benveniste, Beno{\^\i}t Caillaud, Dejan Nickovic, Roberto Passerone,
  Jean-Baptiste Raclet, Philipp Reinkemeier, Alberto Sangiovanni-Vincentelli,
  Werner Damm, Thomas~A. Henzinger, and Kim~G. Larsen.
\newblock {Contracts for System Design}.
\newblock \emph{Foundations and Trends in Electronic Design Automation},
  12\penalty0 (2--3):\penalty0 124--400, 2018.

\bibitem[Hoare(1969)]{hoare1969axiomatic}
C.~A.~R. Hoare.
\newblock {An Axiomatic Basis for Computer Programming}.
\newblock \emph{Communications of the ACM}, 12\penalty0 (10):\penalty0
  576--580, 583, 1969.
\newblock \url{https://doi.org/10.1145/363235.363259}.

\bibitem[Mavin et~al.(2009)Mavin, Wilkinson, Harwood, and Novak]{mavin2009ears}
Alistair Mavin, Philip Wilkinson, Adrian Harwood, and Mark Novak.
\newblock {Easy Approach to Requirements Syntax (EARS)}.
\newblock In \emph{17th IEEE International Requirements Engineering
  Conference}, pages 317--322, 2009.
\newblock \url{https://doi.org/10.1109/RE.2009.9}.

\bibitem[Clements et~al.(2010)Clements, Bachmann, Bass, Garlan, Ivers, Little,
  Merson, Nord, and Stafford]{clements2010architectures}
Paul Clements, Felix Bachmann, Len Bass, David Garlan, James Ivers, Reed
  Little, Paulo Merson, Robert Nord, and Judith Stafford.
\newblock \emph{{Documenting Software Architectures: Views and Beyond}}.
\newblock Addison-Wesley, second edition, 2010.

\bibitem[van Rossum et~al.(2014)van Rossum, Lehtosalo, and
  Langa]{vanrossum2014pep484}
Guido van Rossum, Jukka Lehtosalo, and {\L}ukasz Langa.
\newblock {PEP 484---Type Hints}.
\newblock \url{https://peps.python.org/pep-0484/}, 2014.

\bibitem[{Qwen Team}(2026)]{qwen2026qwen38}
{Qwen Team}.
\newblock {Qwen3.8-27B Model Card}.
\newblock \url{https://huggingface.co/Qwen/Qwen3.8-27B}, 2026.

\bibitem[{Z.AI}(2026)]{zai2026glm52}
{Z.AI}.
\newblock {GLM-5.2: Overview}.
\newblock \url{https://docs.z.ai/guides/llm/glm-5.2}, 2026.

\bibitem[{Strands Labs}()]{strandsaifunctions}
{Strands Labs}.
\newblock {Strands AI Functions}.
\newblock \url{https://github.com/strands-labs/ai-functions}.
\newblock Python library.

\bibitem[McCabe(1976)]{mccabe1976complexity}
Thomas~J. McCabe.
\newblock {A Complexity Measure}.
\newblock \emph{IEEE Transactions on Software Engineering}, SE-2\penalty0
  (4):\penalty0 308--320, 1976.

\bibitem[Endres et~al.(2024)Endres, Fakhoury, Chakraborty, and
  Lahiri]{endres2024nl2postcond}
Madeline Endres, Sarah Fakhoury, Saikat Chakraborty, and Shuvendu~K. Lahiri.
\newblock {Can Large Language Models Transform Natural Language Intent into
  Formal Method Postconditions?}
\newblock \emph{Proceedings of the ACM on Software Engineering}, 1\penalty0
  (FSE), 2024.

\bibitem[Zelikman et~al.(2023)Zelikman, Huang, Poesia, Goodman, and
  Haber]{zelikman2023parsel}
Eric Zelikman, Qian Huang, Gabriel Poesia, Noah~D. Goodman, and Nick Haber.
\newblock {Parsel: Algorithmic Reasoning with Language Models by Composing
  Decompositions}.
\newblock In \emph{Advances in Neural Information Processing Systems}, 2023.

\bibitem[Hong et~al.(2024)Hong, Zhuge, Chen, Zheng, Cheng, Zhang, Wang, Wang,
  Yau, Lin, Zhou, Ran, Xiao, Wu, and Schmidhuber]{hong2024metagpt}
Sirui Hong, Mingchen Zhuge, Jonathan Chen, Xiawu Zheng, Yuheng Cheng, Ceyao
  Zhang, Jinlin Wang, Zili Wang, Steven Ka~Shing Yau, Zijuan Lin, Liyang Zhou,
  Chenyu Ran, Lingfeng Xiao, Chenglin Wu, and J{\"u}rgen Schmidhuber.
\newblock {MetaGPT: Meta Programming for a Multi-Agent Collaborative
  Framework}.
\newblock In \emph{International Conference on Learning Representations}, 2024.

\bibitem[Jiang et~al.(2024)Jiang, Dong, Wang, Fang, Shang, Li, Jin, and
  Jiao]{jiang2024selfplanning}
Xue Jiang, Yihong Dong, Lecheng Wang, Zheng Fang, Qiwei Shang, Ge~Li, Zhi Jin,
  and Wenpin Jiao.
\newblock {Self-Planning Code Generation with Large Language Models}.
\newblock \emph{ACM Transactions on Software Engineering and Methodology},
  33\penalty0 (7):\penalty0 1--30, 2024.
\newblock \doi{10.1145/3672456}.

\bibitem[Bairi et~al.(2024)Bairi, Sonwane, Kanade, {Vageesh D C}, Iyer,
  Parthasarathy, Rajamani, Ashok, and Shet]{bairi2024codeplan}
Ramakrishna Bairi, Atharv Sonwane, Aditya Kanade, {Vageesh D C}, Arun Iyer,
  Suresh Parthasarathy, Sriram Rajamani, B.~Ashok, and Shashank Shet.
\newblock {CodePlan: Repository-level Coding using LLMs and Planning}.
\newblock \emph{Proceedings of the ACM on Software Engineering}, 1\penalty0
  (FSE):\penalty0 675--698, 2024.
\newblock \doi{10.1145/3643757}.

\bibitem[Luo et~al.(2025)Luo, Zhang, Liu, Wu, Liu, Huang, Huang, Yin, Xin,
  Zhan, Sun, Chen, Li, and Yang]{luo2025rpg}
Jane Luo, Xin Zhang, Steven Liu, Jie Wu, Jianfeng Liu, Yiming Huang, Yangyu
  Huang, Chengyu Yin, Ying Xin, Yuefeng Zhan, Hao Sun, Qi~Chen, Scarlett Li,
  and Mao Yang.
\newblock {RPG: A Repository Planning Graph for Unified and Scalable Codebase
  Generation}.
\newblock \emph{arXiv preprint arXiv:2509.16198}, 2025.
\newblock \url{https://arxiv.org/abs/2509.16198}.

\bibitem[Ruan et~al.(2025)Ruan, Zhang, and Roychoudhury]{ruan2025specrover}
Haifeng Ruan, Yuntong Zhang, and Abhik Roychoudhury.
\newblock {SpecRover: Code Intent Extraction via LLMs}.
\newblock In \emph{Proceedings of the 47th IEEE/ACM International Conference on
  Software Engineering (ICSE)}, pages 963--974, 2025.
\newblock \doi{10.1109/ICSE55347.2025.00080}.

\bibitem[Lahiri et~al.(2022)Lahiri, Fakhoury, Naik, Sakkas, Chakraborty,
  Musuvathi, Choudhury, von Veh, Inala, Wang, and Gao]{lahiri2022ticoder}
Shuvendu~K. Lahiri, Sarah Fakhoury, Aaditya Naik, Georgios Sakkas, Saikat
  Chakraborty, Madanlal Musuvathi, Piali Choudhury, Curtis von Veh,
  Jeevana~Priya Inala, Chenglong Wang, and Jianfeng Gao.
\newblock {Interactive Code Generation via Test-Driven User-Intent
  Formalization}.
\newblock \emph{arXiv preprint arXiv:2208.05950}, 2022.

\bibitem[Fakhoury et~al.(2024)Fakhoury, Naik, Sakkas, Chakraborty, and
  Lahiri]{fakhoury2024ticoder}
Sarah Fakhoury, Aaditya Naik, Georgios Sakkas, Saikat Chakraborty, and
  Shuvendu~K. Lahiri.
\newblock {LLM-Based Test-Driven Interactive Code Generation: User Study and
  Empirical Evaluation}.
\newblock \emph{IEEE Transactions on Software Engineering}, 50\penalty0
  (9):\penalty0 2254--2268, 2024.

\bibitem[Chen et~al.(2023)Chen, Zhang, Nguyen, Zan, Lin, Lou, and
  Chen]{chen2023codet}
Bei Chen, Fengji Zhang, Anh Nguyen, Daoguang Zan, Zeqi Lin, Jian-Guang Lou, and
  Weizhu Chen.
\newblock {CodeT: Code Generation with Generated Tests}.
\newblock In \emph{International Conference on Learning Representations}, 2023.
\newblock \url{https://openreview.net/forum?id=ktrw68Cmu9c}.

\bibitem[Kushnir et~al.(2026)Kushnir, Noorbakhsh, Sreedhar, Cheng, Liu,
  Ranganathan, Alizadeh, Kjolstad, and Subramanian]{kushnir2026designdocs}
Samuel Kushnir, Kimia Noorbakhsh, Kavya Sreedhar, Liqun Cheng, Ming Liu,
  Parthasarathy Ranganathan, Mohammad Alizadeh, Fred Kjolstad, and Suvinay
  Subramanian.
\newblock {Design Docs Are All You Need: An AI-native Machine-Learning
  Performance Tool}.
\newblock \emph{arXiv preprint arXiv:2609.05364}, 2026.
\newblock \url{https://arxiv.org/abs/2609.05364}.

\bibitem[He et~al.(2026)He, Lu, Ma, Ye, Wu, Liu, Bu, and Li]{he2026vibetospec}
Ruirui He, Hanqing Lu, Lezhi Ma, Kai Ye, Qiong Wu, Shangqing Liu, Lei Bu, and
  Xuandong Li.
\newblock {From Vibe to Spec: An Empirical Study to Demystify Specifications in
  Code Generation}.
\newblock In \emph{Proceedings of the 41st IEEE/ACM International Conference on
  Automated Software Engineering}, pages 2780--2792. Association for Computing
  Machinery, 2026.
\newblock \url{https://doi.org/10.1145/3832783.3837515}.

\bibitem[Chen et~al.(2026)Chen, Chang, Lin, Chawa, Chen, Wang, and
  Hassan]{chen2026specfirst}
Yihao Chen, Shi Chang, Feng Lin, Khaled Chawa, Boyuan Chen, Shaowei Wang, and
  Ahmed~E. Hassan.
\newblock {SpecFirst: Behavioral Specification Elicitation as a First-Class
  Step in Agent-Based Program Synthesis from Scratch}.
\newblock \emph{arXiv preprint arXiv:2607.27167}, 2026.
\newblock \url{https://arxiv.org/abs/2607.27167}.

\bibitem[Sun et~al.(2024)Sun, Sheng, Padon, and Barrett]{sun2024clover}
Chuyue Sun, Ying Sheng, Oded Padon, and Clark Barrett.
\newblock {Clover: Closed-Loop Verifiable Code Generation}.
\newblock In \emph{International Symposium on AI Verification}, 2024.

\bibitem[Solar-Lezama(2009)]{solarlezama2009sketching}
Armando Solar-Lezama.
\newblock {The Sketching Approach to Program Synthesis}.
\newblock In \emph{Programming Languages and Systems (APLAS 2009)}, volume 5904
  of \emph{Lecture Notes in Computer Science}, pages 4--13. Springer, 2009.
\newblock \doi{10.1007/978-3-642-10672-9_3}.

\bibitem[Poesia et~al.(2022)Poesia, Polozov, Le, Tiwari, Soares, Meek, and
  Gulwani]{poesia2022synchromesh}
Gabriel Poesia, Oleksandr Polozov, Vu~Le, Ashish Tiwari, Gustavo Soares,
  Christopher Meek, and Sumit Gulwani.
\newblock {Synchromesh: Reliable Code Generation from Pre-trained Language
  Models}.
\newblock In \emph{International Conference on Learning Representations}, 2022.
\newblock \url{https://openreview.net/forum?id=KmtVD97J43e}.

\bibitem[Aggarwal et~al.(2024)Aggarwal, Parno, and
  Welleck]{aggarwal2024alphaverus}
Pranjal Aggarwal, Bryan Parno, and Sean Welleck.
\newblock {AlphaVerus: Bootstrapping Formally Verified Code Generation through
  Self-Improving Translation and Treefinement}.
\newblock \emph{arXiv preprint arXiv:2412.06176}, 2024.
\newblock \url{https://arxiv.org/abs/2412.06176}.

\bibitem[Yang et~al.(2025)Yang, Li, Misu, Yao, Cui, Gong, Hawblitzel, Lahiri,
  Lorch, Lu, Yang, Zhou, and Lu]{yang2025autoverus}
Chenyuan Yang, Xuheng Li, Md~Rakib~Hossain Misu, Jianan Yao, Weidong Cui, Yeyun
  Gong, Chris Hawblitzel, Shuvendu Lahiri, Jacob~R. Lorch, Shuai Lu, Fan Yang,
  Ziqiao Zhou, and Shan Lu.
\newblock {AutoVerus: Automated Proof Generation for Rust Code}.
\newblock \emph{Proceedings of the ACM on Programming Languages}, 9\penalty0
  (OOPSLA2):\penalty0 3454--3482, 2025.
\newblock \doi{10.1145/3763174}.

\bibitem[Bursuc et~al.(2025)Bursuc, Ehrenborg, Lin, Astefanoaei, Chiosa,
  Kukovec, Singh, Butterley, Bizid, Dougherty, Zhao, Tan, and
  Tegmark]{bursuc2025vericoding}
Sergiu Bursuc, Theodore Ehrenborg, Shaowei Lin, Lacramioara Astefanoaei,
  Ionel~Emilian Chiosa, Jure Kukovec, Alok Singh, Oliver Butterley, Adem Bizid,
  Quinn Dougherty, Miranda Zhao, Max Tan, and Max Tegmark.
\newblock {A Benchmark for Vericoding: Formally Verified Program Synthesis}.
\newblock \emph{arXiv preprint arXiv:2509.22908}, 2025.
\newblock \url{https://arxiv.org/abs/2509.22908}.

\bibitem[Chen et~al.(2021)]{chen2021codex}
Mark Chen et~al.
\newblock {Evaluating Large Language Models Trained on Code}.
\newblock \emph{arXiv preprint arXiv:2107.03374}, 2021.
\newblock \url{https://arxiv.org/abs/2107.03374}.

\end{thebibliography}
\end{document}